\documentclass[11pt]{article} 
\usepackage[numbers]{natbib}
\usepackage{booktabs} 
\usepackage{xcolor}
\usepackage[utf8]{inputenc} 

\usepackage{pstricks}
\usepackage{comment}
\usepackage{dsfont}
\usepackage{tikz}
\usepackage{tikz-dependency}
\usepackage{mathrsfs}
\usepackage{appendix}

\usepackage{amsmath,amsthm,amsfonts,latexsym,bbm,xspace,graphicx,float,mathtools,mathdots}
\usepackage{braket,caption,subcaption,ellipsis,xcolor,textcomp}
\usepackage{enumitem}
\usepackage[ruled]{algorithm2e}
\IncMargin{-\parindent}
\usepackage{algpseudocode}
\usepackage{listings}
\usepackage{algcompatible}
\usepackage{float}
\floatstyle{ruled}

\usepackage[colorlinks,allcolors=teal]{hyperref}
\usepackage{thmtools}
\usepackage{thm-restate}
\usepackage{cleveref}

\def\showauthnotes{1}

\ifthenelse{\showauthnotes=1}
{
\newcommand{\authnote}[2]{[{\color{blue}\textbf{#1:}}~{\color{blue} #2}]}
}
{ \newcommand{\authnote}[2]{} 
}

\usepackage{accents}

\newcommand{\alloc}{X}
\newcommand{\ralloc}{\mathbf{R}}
\newcommand{\allocy}{Y}

\newcommand{\bids}{\mathbf{b}}
\newcommand{\values}{\mathbf{v}}
\newcommand{\items}{\mathrm{M}}
\newcommand{\agents}{\mathrm{N}}

\newcommand{\NOM}{\textsc{NOM}}

\newcommand{\EFone}{\mathrm{EF}1}
\newcommand{\Alloc}{A}
\newcommand{\SW}{\textsc{SW}}
\newcommand{\NSW}{\textsc{NSW}}

\newcommand{\ESW}{\textsc{ESW}}

\newtheorem{property}{Property}

\DeclareMathOperator*{\argmax}{arg\,max}

\input{stylesheet.sty}

\usepackage{pifont} 
\newcommand{\cmark}{\ding{51}}
\newcommand{\xmark}{\ding{55}}
\usepackage{tikz}
\usetikzlibrary{decorations.pathreplacing,calc}
\newcommand{\tikzmark}[1]{\tikz[overlay,remember picture] \node (#1) {};}
\newcommand*{\AddNote}[4]{%
     \begin{tikzpicture}[overlay, remember picture]
         \draw [decoration={brace,amplitude=0.5em},decorate,thick,red]
             ($(#3)!(#1)!($(#3)-(0,1)$)$) --  
             ($(#3)!(#2)!($(#3)-(0,1)$)$)
                 node [align=center, text width=1.5cm, pos=0.5, anchor=west] {#4};
     \end{tikzpicture}
     }
\usepackage{authblk}

\begin{document}
\title{Fair and Efficient Allocations Without Obvious Manipulations\thanks{The first author is supported in part by a Google Research Scholar Award.}}
\author[a]{Alexandros Psomas}
\author[a]{Paritosh Verma}
\affil[a]{Department of Computer Science, Purdue University.

Email: \texttt{{\{apsomas,verma136\}@cs.purdue.edu}}}
\date{} 

\maketitle

\begin{abstract}
We consider the fundamental problem of allocating a set of indivisible goods among strategic agents with additive valuation functions. It is well known that, in the absence of monetary transfers, Pareto efficient and truthful rules are dictatorial, while there is no deterministic truthful mechanism that allocates all items and achieves envy-freeness up to one item (EF1), even for the case of two agents. In this paper, we investigate the interplay of fairness and efficiency under a relaxation of truthfulness called non-obvious manipulability (NOM), recently proposed by~\cite{troyan2020obvious}. We show that this relaxation allows us to bypass the aforementioned negative results in a very strong sense. Specifically, we prove that there are deterministic and EF1 algorithms that are not obviously manipulable, and the algorithm that maximizes utilitarian social welfare (the sum of agents' utilities), which is Pareto efficient but not dictatorial, is not obviously manipulable for $n \geq 3$ agents (but obviously manipulable for $n=2$ agents). At the same time, maximizing the egalitarian social welfare (the minimum of agents' utilities) or the Nash social welfare (the product of agents' utilities) is obviously manipulable for any number of agents and items. Our main result is an approximation preserving black-box reduction from the problem of designing EF1 and NOM mechanisms to the problem of designing EF1 algorithms. En route, we prove an interesting structural result about EF1 allocations, as well as new ``best-of-both-worlds'' results (for the problem without incentives), that might be of independent interest.
\end{abstract}

\section{Introduction}\label{sec: intro}

We consider the fundamental problem of allocating a set of indivisible items among strategic agents with additive valuation functions.
It is well-understood that, in the absence of monetary transfers, fairness, efficiency and truthfulness cannot be reconciled, in a very strong sense.
For example, a serial dictatorship (arguably the most unfair allocation rule) is the \emph{unique} truthful and Pareto efficient mechanism~\cite{papai2000strategyproof, klaus2002strategy, ehlers2003coalitional}, even for the case of two agents and randomized mechanisms (or, equivalently, divisible items)~\cite{schummer1996strategy}. At the same time, achieving envy-freeness up to one item (EF1)~\cite{budish2011combinatorial}, one of the most popular fairness notions for indivisible goods that is compatible with Pareto efficiency~\cite{caragiannis2019unreasonable}, is impossible for truthful and deterministic mechanisms to achieve, even for two additive agents~\cite{amanatidis2017truthful}.

The standard abstraction for a strategic agent assumes that, when presented with an algorithm, the agent will perfectly understand it and optimally respond to it. Therefore, a truthful mechanism should protect against all deviating strategies. In this paper we aim to escape the aforementioned impossibility results for truthful mechanisms by relaxing this requirement: our goal is to design mechanisms that protect only against \emph{obviously dominant} deviating strategies. 

Informally, a strategy is obviously dominant if it guarantees an outcome better than any outcome of every other strategy~\cite{li2017obviously}. Designing obviously truthful mechanisms, those where honest reporting is obviously dominant, is a problem that has attracted significant attention in the past few years~\cite{ashlagi2018stable,kyropoulou2019obviously,ferraioli2019obvious,arribillaga2020obvious}. However, since dishonesty is also a strategy, a deviation can be profitable but not obviously profitable. A manipulation is obvious if it yields a higher utility than truth-telling in either the best or worst-case scenarios. Seminal work of~\cite{troyan2020obvious} defines this notion, and shows that relaxing truthfulness to non-obvious manipulability can strictly improve the designer's objective in a number of domains (e.g. two-sided matching).

In this paper we aim to provide a comprehensive study of what can and what cannot be achieved by non-obviously manipulable mechanisms in the fundamental problem of fair and efficient allocation of indivisible goods.

\subsection{Our Contribution}

We first consider the problem of designing a non-obviously manipulable mechanism (NOM) which always outputs allocations that are envy-free up to one item (EF1), a task that is impossible for truthful mechanisms to achieve~\cite{amanatidis2017truthful}. We prove that the Round-Robin procedure (agents choose items one at a time, following a fixed order) which is known to always output EF1 allocations, is also not obviously manipulable (Theorem~\ref{thm:rr is nom}), giving a separation between truthfulness and non-obviously manipulability for our problem. In fact, we show that NOM is compatible with even stronger notions of fairness. Specifically, a fairness guarantee stronger than EF1 that a randomized algorithm could satisfy is to be ex-post EF1 (i.e. to always output an EF1 allocation), and simultaneously be envy-free in expectation (ex-ante). There are known algorithms that satisfy this ``best-of-both-worlds'' guarantee~\cite{freeman2020best,aziz2020simultaneously}; we show that the PS-Lottery algorithm of~\cite{aziz2020simultaneously} is not-obviously manipulable. This result uses a connection between ex-ante proportionality and not obvious manipulability; we later exploit this connection to give new negative ``best-of-both-worlds'' results!

We proceed to study efficient algorithms. Specifically, we study the three most common notions of efficiency in this setting: utilitarian social welfare (the sum of agents' utilities), egalitarian social welfare (the minimum of agents' utilities), and Nash social welfare (the geometric mean of agents' utilities). For the case of $n=2$ agents all three notions are incompatible with non obvious manipulability, i.e. every (deterministic or randomized) algorithm that always outputs an integral solution that is optimal with respect to any of these objectives is obviously manipulable (Theorem~\ref{theorem:NOM+SWM+2agents}). This is also true for the case of more than three agents for egalitarian and Nash social welfare (Theorems~\ref{thm: NOM + egalitarian + n agents} and~\ref{thm: MNW impossibility}).
The intuition for these impossibilities is the following. There are instances where a specific agent\footnote{Informally, the agent that the tie-breaking rule favors.} gets her least favorite item, in the worst-case. However, solutions that are optimal for egalitarian and Nash welfare prioritize maximizing the number of agents that get non-zero utility. Therefore, by reporting that she only values a single item (her favorite item) to the mechanism, this agent can force the mechanism to give her this item (or increase the probability that she gets it) in the worst-case.
Surprisingly, the same is not true for utilitarian social welfare:\footnote{Note that when discussing utilitarian social welfare, we assume that the input is normalized, i.e. the utility of an agent for the grand bundle of items is equal to $1$. See Section~\ref{sec: efficient} for details.} there is a not-obviously manipulable algorithm that always outputs an allocation that maximizes utilitarian social welfare for the case of $n \geq 3$ agents (Theorem~\ref{theorem:NOM+SWM+nagents}). This gives an efficient and not obviously manipulable mechanism that is not dictatorial.

Since maximizing Nash social welfare, a common way to simultaneously achieve Pareto efficiency and envy-freeness up to one item, is obviously manipulable, the next natural question is whether there are any other Pareto efficient and EF1 algorithms that are not-obviously manipulable. Our main result answers this question in the affirmative. In fact, we prove that the problem of designing a Pareto efficient, EF1 and NOM \emph{mechanism} is exactly as hard as designing a Pareto efficient and EF1 \emph{algorithm}. We give a black-box reduction, that preserves Pareto efficiency guarantees, which given an algorithm that always outputs (clean and non-wasteful) EF1 allocations, outputs a mechanism that is not obviously manipulable and always outputs EF1 allocations.
There are two crucial steps in our reduction. Given the valuation function $\values_i$ of agent $i$, let $\EFone(i,\values_i)$ be the set of allocations that are clean, non-wasteful and $\EFone$ for agent $i$. 
First, our reduction ensures that, when agent $i$ considers a reporting her valuation as $\bids_i$, every single allocation in $\EFone(i,\bids_i)$ is possible, i.e. for every allocation $\Alloc \in \EFone(i,\bids_i)$ there are reports for the other agents such that $\Alloc$ is the output of the mechanism. Second, we prove the following intuitive structural result, that might be of independent interest. For an agent $i$ with valuation function $\values_i$, the worst-case allocation in $\EFone(i,\values_i)$ is better than the worst-case allocation in $\EFone(i,\bids_i)$ (Lemma~\ref{lemma:least-valued-EF1}). These facts combined give the ``worst-case'' part from the definition of NOM, which is the most challenging step. 
As a direct application of our reduction, by giving as input the algorithmic results~\cite{garg2021fair}, we get a fractionally Pareto efficient, EF1 and NOM mechanism for additive agents; this last result cannot be improved by adding ex-ante fairness guarantees (due to a result of~\cite{freeman2020best}).

Our results show a connection between certain notions of fairness and not obvious manipulability: positive algorithmic results for $\EFone$ and ``best-of-both-worlds'' guarantees can be used to get positive results for not obviously manipulable mechanisms. In Section~\ref{sec: best of both} we exploit this connection in the other direction, and show how negative results for not obviously manipulable mechanisms can be used to give new negative results for ``best-of-both-worlds'' algorithms. 
Specifically, we prove that it is impossible to achieve ex-ante proportionality, ex-post Pareto efficiency while ex-post maximizing the number of agents with positive utility. As direct corollaries we recover a known result of~\cite{freeman2020best} that ex-ante proportional and ex-post MNW allocations do not exist, as well as two new results: it is impossible to achieve ex-ante proportionality by randomizing over allocations that are (1) ex-post Pareto efficiency and ex-post egalitarian, or (2) ex-post leximin.

\begin{center}
\begin{table}
\centering
\begin{tabular}{| c | p{50mm} |}
\hline
 Fairness \& Efficiency Notions & Existence of NOM Mechanism  \\ 
 \hline \hline
 EF1 & \cmark [Theorem \ref{thm:rr is nom}]  \\ 
 \hline
 Ex-ante EF \& Ex-post EF1 & \cmark  [Theorem \ref{best-of-both-worlds-and-NOM}] \\
 \hline
 Social Welfare Maximization & \xmark $\ n=2$ agents [Theorem \ref{theorem:NOM+SWM+2agents}] \newline \cmark $n \geq 3$ agents [Theorem \ref{theorem:NOM+SWM+nagents}] \\
 \hline
 Egalitarian Welfare Maximization & \xmark [Theorem \ref{thm: NOM + egalitarian + n agents}] \\
 \hline
 Nash Social Welfare Maximization & \xmark [Theorem \ref{thm: MNW impossibility}] \\
 \hline
 fPO + EF1 & \cmark [Theorem \ref{mechanism:NOM+PO+EF1}] \\
 \hline
 Ex-ante EF \& Ex-post fPO + EF1 & \xmark ~\cite{freeman2020best} \\
 \hline
\end{tabular}
\caption{\label{tab:table-name}Existence of NOM mechanisms for various fairness and efficiency notions.}
\end{table}
\end{center}

\subsection{Related work}

\cite{troyan2020obvious} define the notion of not obvious manipulability. They show that a number of known mechanisms which are not strategyproof satisfy not obviously manipulability, e.g. stable mechanisms in the context of two-sided matching are not obviously manipulable (in direct contrast to~\cite{roth1982economics}, who shows that there exists no mechanism that is both stable and strategyproof). At the same time, many known not strategyproof mechanisms are obviously manipulable, e.g. the first-price auction, and more generally the pay-as-bid auction.


The notion of non obvious manipulability in computational social choice has been explored in some very recent works~\cite{ortega2019obvious,aziz2021obvious}.~\cite{ortega2019obvious} study cake-cutting. They show that, as opposed to truthfulness, NOM is compatible with proportionality: an adaptation of the well-known moving knife procedure satisfies both properties. They also observe that every proportional (direct revelation) rule satisfies the worst-case guarantee required for NOM (we show a similar statement here, in Lemma~\ref{lemma:ex-ante-EF-and-NOM}), while every Pareto efficient (direct revelation) rule satisfies the best-case guarantee required for NOM. Therefore, in the context of direct revelation mechanisms in cake-cutting, proportional and Pareto efficient rules are not obviously manipulable. \cite{aziz2021obvious} explore obvious manipulations in voting theory. They give sufficient conditions for voting rules to be not obviously manipulable, and show that a number of voting rules escape the pessimistic conclusions of the Gibbard-Satterthwaite theorem, e.g., the Borda rule is not obviously manipulable. After the conference version of this paper was published, several new works study non-obvious manipulability in extensive-form mechanisms~\cite{archbold2023non}, matching~\cite{arribillaga2023obvious}, bilateral trade~\cite{archbold2023nontrade}, and single-peaked preferences~\cite{arribillaga2023not}.

In the problem we study here, one can bypass the aforementioned impossibility results for no-money truthful mechanisms by relaxing the notion of incentive compatibility to Bayesian incentive compatibility (BIC)~\cite{gkatzelis2023getting}.
Other than relaxing the notion of incentive compatibility, other works restrict agents' valuations, e.g. focus on dichotomous~\cite{halpern2020fair,babaioff2021fair,benabbou2021finding,barman2021truthful} or 
Leontief valuations~\cite{ghodsi2011dominant,friedman2014strategyproof,parkes2015beyond}, or use money-burning (wasting resources) as a substitute for payments~\cite{hartline2008optimal,cole2013mechanism,fotakis2016efficient,friedman2019fair,abebe2020truthful}.

Ignoring incentives altogether, there is a growing literature that studies the problem of allocating indivisible goods in a fair and efficient manner.
For additive agents,~\cite{caragiannis2019unreasonable} show that the integral allocation that maximizes the Nash social welfare is Pareto efficient and envy-free up to one item. Finding such allocations is APX-hard~\cite{lee2017apx}. However,~\cite{barman2018finding} give an algorithm for finding a Pareto efficient and $\EFone$ allocation that runs in pseudo-polynomial time (and polynomial time for bounded valuations) for the case of additive agents.~\cite{barman2018finding} also provide a proof of existence for $\EFone$ and fractionally Pareto efficient allocations; recent work by~\cite{garg2021fair} gives a pseudo-polynomial time algorithm for finding such allocations. Beyond additive agents,~\cite{caragiannis2019unreasonable} show that there are no Pareto efficient and $\EFone$ allocations for subadditive or supermodular valuations. At the same time, the problem of approximating the Nash social welfare under different valuation functions (sometimes even combined with some fairness guarantees) has received significant attention in the past few years~\cite{cole2018approximating,cole2017convex,anari2018nash,garg2018approximating,chaudhury2018fair,chaudhury2021fair,chaudhury2021little}. 

\section{Preliminaries}\label{sec: prelims}

We consider the problem of allocating a set $\items$ of $m$ items among a set $\agents$ of $n$ agents with additive utilities. Henceforth we will use $[k]$ to denote the set $\{1,2, \ldots, k\}$ for any positive integer $k \in \mathbb{Z}_+$, and we will assume that $\items = [m]$ and $\agents = [n]$.


\paragraph{Allocations}
A \emph{fractional allocation} $A \in [0,1]^{n\cdot m}$ is a $m \times n$ matrix that defines for each agent $i \in \agents$ and item $j \in \items$ the fraction $A_{i,j}$ of the item $j$ that the agent $i$ will receive. We represent a fractional allocation as $A = (A_1, A_2, \ldots, A_n)$ where $A_i = (A_{i,1}, A_{i,2}, \ldots, A_{i,m}) \in [0,1]^m$ denotes the fractions of all items allocated to agent $i \in \agents$. A feasible allocation satisfies $\sum_{i \in \agents} A_{i,j} \leq 1$, for all items $j \in \items$. 
A fractional allocation $A$ is \emph{integral} if $A_i \in \{ 0, 1 \}^m$ for all agents $i \in \agents$. 
An integral allocation can be equivalently defined as $n$ disjoint subsets of the set of items $\items$, this representation is often convenient. Let $\Pi_n(\items) = \{(S_1, S_2, \ldots, S_n) \ | \ \cup_{i=1}^n S_i \subseteq M \text{ and } \forall i \neq j,\ S_i \cap S_j = \emptyset \}$ denote the set of all $n$ ordered disjoint subsets of the set of items $\items$. Given an integral allocation $A = (A_1, \ldots, A_n)$, we can interpret the binary vectors $A_i = (A_{i,1}, A_{i,2}, \ldots, A_{i,m})$ as subsets of items $A_i \coloneqq \{j \in \items \ | \ A_{i,j} = 1\}$. This allows us to view an integral allocation $A$ as $n$ ordered disjoint subsets of items, i.e., $A = (A_1, A_2, \ldots, A_n) \in \Pi_n(\items)$. An integral allocation $A = (A_1, A_2, \ldots, A_n)$ is \emph{complete} if $\cup_{i=1}^n A_i = M$ and \emph{partial} if $\cup_{i=1}^n A_i \subset M$. Unless stated otherwise, we use the term allocation to refer to complete allocations. We use the term \emph{bundle} to refer to any subset of items.


A \emph{randomized allocation} $\ralloc = \{(p^z, A^z)\}_{z=1}^k$ is a probability distribution (or a lottery) over a set of integral allocations, i.e., for every $z \in [k]$, $A^z$ is an integral allocation that occurs with a probability $p^z$. The sum of the probabilities is equal to one, $\sum_{z=1}^k p^z = 1$. The integral allocations $A^1, A^2, \ldots, A^k$ constitute the \emph{support} of $\ralloc$. For each randomized allocation $\ralloc = \{(p^z, A^z)\}_{z=1}^k$ there exists a corresponding \emph{expected fractional allocation} $X = \sum_{z=1}^k p^z \cdot A^z$, where for each agent $i \in \agents$, the fraction of items given to agent $i$ is $X_i = \sum_{z=1}^k p^z \cdot A^z_i$.\footnote{Note that multiple randomized allocations may have the same expected fractional allocation.} 
Here $X_{i,j}$ can be thought of as the probability with which agent $i$ is allocated item $j$ in an integral allocation that is sampled from $\ralloc$. For notational clarity, we will use the letters $X,Y$ to denote allocations that are fractional or integral, $A,B$ to denote allocations that are exclusively integral and $\ralloc$ for randomized allocations.


\paragraph{Preferences}
Each agent $i \in \agents$ has a private valuation function $\values_i(.)$ that outputs the utility that agent $i$ derives from a given set (or fractions) of items. We use the notation $\values_i(X_i)$ (resp. $\values_i(A_i)$) to denote the utility that agent $i$ gets from the items allocated to her in a fractional allocation $X$ (resp. integral allocation $A$).
We consider agents with additive utilities. An additive agent $i \in \agents$ has a non-negative valuation $v_{i,j}$ for receiving the entirety of item $j$; her utility for an allocation $X = (X_1, \dots, X_n)$ is $\values_i(X_i) = \sum_{j \in \items} X_{i,j} v_{i,j}$; for an integral allocation $A = (A_1, A_2, \ldots, A_n)$, the utility is simply $\values_i(A_i) = \sum_{j \in A_i} v_{i,j}$. 

\paragraph{Mechanisms}
A mechanism $\mathcal{M}$ elicits ``bids'' (i.e. reported valuations) $\bids = (\bids_1, \dots, \bids_n)$ from every agent $i \in \agents$, and outputs a feasible allocation. We write $X_{i,j} (\bids)$ for the fraction of item $j$ allocated to agent $i$ when each agent $j \in \agents$ reports a valuation $\bids_j$. 
We use the notation $\bids$ (and $\bids_i$) to refer to the input to a mechanism, and $\values$ (and $\values_i$) to refer to the true valuations of agents.

A mechanism $\mathcal{M}$ is \emph{deterministic} if for every reported valuations $\bids$ it deterministically outputs an integral allocation. We will use $\mathcal{M}(\bids) = (\mathcal{M}_1(\bids),\mathcal{M}_2(\bids), \ldots, \mathcal{M}_n(\bids))$ to denote the integral allocation that a mechanism $\mathcal{M}$ outputs given reported valuations $\bids = (\bids_1, \dots, \bids_n)$ as input, here $\mathcal{M}_i(\bids) \subseteq \items$ represents the bundle of items that agent $i$ receives in the allocation $\mathcal{M}(\bids)$. A mechanism $\mathcal{R}$ is \emph{randomized} if for every reported valuations $\bids$ it outputs a randomized allocation, i.e. it returns integral allocations that are drawn from a probability distribution corresponding to a randomized allocation. 
Since every randomized allocation has an associated expected fractional allocation, the output of a randomized mechanism for reported valuations $\bids$ can also be interpreted as representing a fractional allocation. We use $\mathcal{R}(\bids) = (\mathcal{R}_1(\bids), \ldots, \mathcal{R}_n(\bids))$ to denote the \emph{expected fractional allocation} that a randomized mechanism outputs given bids $\bids$; the vector $\mathcal{R}_i(\bids) \in [0,1]^m$ represents the probabilities with which agent $i$ receives each item in the sampled integral allocation.

Note that, for a deterministic mechanism $\mathcal{M}$, the value $\values_i(\mathcal{M}_i(\bids))$ denotes the utility of agent $i$ (as per her true valuation $\values_i$) for the allocation output by $\mathcal{M}$ on input $\bids$. Similarly, $\values_i(\mathcal{R}_i(\bids))$ denotes the expected utility of agent $i$ in the sampled integral allocation when the input to a randomized mechanism $\mathcal{R}$ is $\bids$.



\paragraph{Notions of incentive compatibility}
A mechanism $\mathcal{M}$ is \emph{truthful} if agents cannot strictly improve their utility by misreporting their valuation, i.e. for all $i \in \agents$, valuations $\values_i, \bids_i$, and reports of the other agents $\values_{-i}$, $\values_i( \mathcal{M}_i(\values_i, \values_{-i}) ) \geq \values_i( \mathcal{M}_i(\bids_i, \values_{-i}) )$. Our work focuses on a notion of incentive compatibility that is a relaxation of truthfulness called \emph{not obvious manipulability}.

\begin{definition}[Not Obviously Manipulable~\cite{troyan2020obvious}]\label{dfn: nom}
A 
mechanism $\mathcal{M}$ is \emph{not obviously manipulable} (or simply NOM) if for every agent $i \in \agents$ with valuation function $\values_i$, and every possible report $\bids_i$ of agent $i$, the following two inequalities hold:
\begin{enumerate}[start = 1, label={(\arabic*)}]
    \item\label{inequality:NOM-worst-case} $\min\limits_{\values_{-i}} \values_i( \mathcal{M}_i(\values_i, \values_{-i}) ) \geq \min\limits_{\values_{-i}} \values_i( \mathcal{M}_i(\bids_i, \values_{-i}) )$.
    \item\label{inequality:NOM-best-case} $\max\limits_{\values_{-i}} \values_i( \mathcal{M}_i(\values_i, \values_{-i}) ) \geq \max\limits_{\values_{-i}} \values_i( \mathcal{M}_i(\bids_i, \values_{-i}) )$.
\end{enumerate}
\end{definition}

Intuitively, if a mechanism is NOM then an agent cannot increase her worst-case utility or increase her best-case utility (computed with respect to the true valuation) by misreporting her valuation, i.e., the best-case and worst-case utilities are maximized when agents report their true valuation. If either the worst-case utility or the best-case utility can be improved then the mechanism is \emph{obviously manipulable}.

To define NOM for a randomized mechanism $\mathcal{R}$, we compare the expected utilities $\values_i(\mathcal{R}_i(\cdot \ , \cdot))$ in inequality \ref{inequality:NOM-worst-case} and \ref{inequality:NOM-best-case}, instead of $\values_i(\mathcal{M}_i(\cdot \ , \cdot))$ as in the case of deterministic mechanisms. Notice that for randomized mechanisms, the definition of NOM takes an expectation over the randomness of the mechanism, and minimum/maximum are over the reports of other agents; we sometimes write ``NOM in expectation'' when referring specifically to a randomized mechanism.

\paragraph{Notions of efficiency}

An integral allocation $\Alloc = (A_1, \ldots, A_n)$ is \emph{Pareto efficient} (or PO) iff there is no integral allocation $\Alloc' = (A'_1, \ldots, A'_n)$ such that for all agents $i \in \agents$, $\values_i(A'_i) \geq \values_i(A_i)$, and for at least one agent this inequality is strict. An (integral or fractional) allocation $X = (X_1, \ldots, X_n)$ is \emph{fractionally Pareto efficient} (or fPO) iff there is no fractional allocation $X' = (X'_1, \ldots, X'_n)$ such that for all agents $i \in \agents$, $\values_i(X'_i) \geq \values_i(X_i)$, and for at least one agent this inequality is strict. Note that fractional Pareto efficiency is a strictly stronger notion than Pareto efficiency. 
An (integral or fractional) allocation $X$ is $\alpha$-approximately (resp. fractionally) Pareto efficient if there is no integral allocation (resp. fractional allocation) $X' = (X'_1, X'_2, \ldots, X'_n)$ such that $\alpha \cdot \values_i(X'_i) \geq \values_i(X_i)$, with at least one of these inequalities strict~\cite{ruhe1990varepsilon,immorlica2017approximate,friedman2019fair,zeng2020fairness}.
Notice that for $\alpha = 1$ we exactly recover Pareto efficiency (resp. fractional Pareto efficiency).

A (partial) allocation $A = (A_1, A_2, \ldots, A_n)$ is \emph{non-wasteful} iff for each $i \in \agents$, $v_{i,j} = 0$ for every unallocated item $j \in \items \setminus \cup_{k \in \agents} A_k$\footnote{Note that, every Pareto efficient (fractionally Pareto efficient) integral allocation is non-wasteful.}. A bundle $S \subseteq \items$ is clean for agent $i$ if the removal of any good reduces $i$'s value for $s$  $\values_i(S) > \values_i (S \setminus \{ g \} )$ for all $g \in S$. Further, a (partial) allocation $A = (A_1, A_2, \ldots, A_n)$ is clean if for each $i \in \agents$ the bundle $A_i$ is clean for agent $i$.

Often, we are interested in computing or approximating specific points (i.e., allocations) present on the Pareto frontier that have additional desirable properties.
The \emph{utilitarian social welfare} of an allocation $\alloc$, denoted by $\SW(\alloc)$, is defined as the sum of utilities that each agent gets in the allocation $\alloc$, i.e., $\SW(\alloc) = \sum_{i \in \agents} \values_i(X_i)$. 

The \emph{Nash social welfare} of an (integral or fractional) allocation $\alloc$, denoted by $\NSW(\alloc)$, is defined as the geometric mean of the utilities of agents in the allocation $\alloc$, i.e., $\NSW(\alloc) = (\prod_{i \in \agents} \values_i(x_i))^{\frac{1}{n}}$. An \emph{integral} allocation that maximizes the Nash social welfare, among all integral allocations, is called an \emph{Nash social welfare maximizing} (or MNW) allocation. There are instances where every integral allocation $\Alloc$ is such that $\NSW(\Alloc) = 0$, i.e., there is always an agent having zero utility. To cover such an edge case, integral MNW allocations are defined as follows \cite{caragiannis2019unreasonable}. An integral allocation is MNW iff $(i)$ it maximizes, among the set of all integral allocations, the number of agents having positive utility and $(ii)$ for any such maximal set of agents $S$, it maximizes the geometric mean of the utilities of agents in $S$.

The \emph{egalitarian social welfare} of an (integral or fractional) allocation $\alloc$, denoted as $\ESW(\alloc)$, is defined as the minimum utility that any agent derives from allocation $\alloc$, i.e., $\ESW(\alloc) = \min_{i \in \agents} \values_i(x_i)$. An integral allocation is \emph{egalitarian social welfare maximizing} iff $(i)$ it maximizes, among the set of all integral allocations, the number of agents having positive utility and $(ii)$ for any such maximal set of agents $S$, it maximizes the egalitarian social welfare, i.e., the minimum utility of agents in $S$.

\paragraph{Notions of fairness}
An (integral or fractional) allocation $\alloc$ is called \emph{envy-free} (EF) if for every pair of agent $i,j \in \agents$, agent $i$ values her allocation at least as much as the allocation of agent $j$, i.e., $\values_i(X_i) \geq \values_i(X_j)$. An (integral or fractional) allocation $\alloc$ is called \emph{proportional} (PROP) if the utility of each agent is at least a $\frac{1}{n}$ fraction of her utility for all the items, i.e., $\values_i(X_i) \geq \frac{1}{n} \values_i(\items)$ for every $i \in \agents$\footnote{Note that, if an (integral or fractional) allocation is envy-free then it is proportional as well.}.
When considering integral allocations, achieving proportionality or envy-freeness is impossible (consider the case of a single item and two agents that both want it), so weaker notions of fairness have been defined. An integral allocation $\Alloc$ is envy-free up to one item (EF1) if for every pair of agents $i,j \in \agents$, where $A_j \neq \emptyset$, agent $i$ values her allocation at least as much as the allocation of agent $j$, subject to the removal of one item from agent $j$'s bundle, i.e., $\values_i(A_i) \geq \values_i(A_j \setminus \{g \})$ for some item $g \in A_j$.

Finally, we are interested in, so called, best-of-both-worlds guarantees. Let $\mathcal{P}$ be a fairness or efficiency notion for integral allocations and $\mathcal{Q}$ be a fairness or efficiency notion for fractional allocations. A randomized allocation $\ralloc = \{(p^z, \Alloc^z)\}_{z=1}^k$ satisfies the notion $\mathcal{P}$ \emph{ex-post} if each integral allocation $\Alloc^z$ in the support of $\ralloc$ satisfies the notion $\mathcal{P}$. Additionally, the randomized allocation $\ralloc$ satisfies the notion $\mathcal{Q}$ \emph{ex-ante} if the expected fractional allocation $\allocy = \sum_{z=1}^k p^z \cdot \Alloc^z$ corresponding to $\ralloc$ satisfies the notion $\mathcal{Q}$.

\section{Fair Mechanisms}\label{sec: fair}
In this section we study whether non obvious manipulability is compatible with envy freeness up to one item. 
For deterministic mechanisms, it is known that there is no truthful mechanism that always outputs an EF1 allocation, even for the case of two agents~\cite{amanatidis2017truthful}. 
In sharp contrast, we show that Round-Robin, arguably the simplest algorithm that guarantees EF1, is not obviously manipulable.
Recall that the Round-Robin algorithm allocates items to agents over a sequence of rounds. In each round, the agents choose one item each (a highest-value remaining item, as per their valuation)\footnote{Ties can be broken in an arbitrary manner.}
following a fixed, arbitrary order.

\begin{theorem}\label{thm:rr is nom}
Round-Robin is \emph{not obviously manipulable}.
\end{theorem}

\begin{proof}
We prove that the two inequalities in Definition~\ref{dfn: nom} hold for Round-Robin. That is, an agent cannot improve her worst-case or best-case utilities by misreporting. Let $i \in \agents$ be the  $i$-th agent in the Round-Robin order, and let $\values_i$ be her valuation vector. Assume without loss of generality that $v_{i,j} \geq v_{i,j+1}$ for all $j = 0, \dots, m - 1$.
Let $\ell$ be the number of items agent $i$ receives (either $\lfloor \frac{n}{m} \rfloor$ or $\lfloor \frac{n}{m} \rfloor + 1$, depending on her rank, $n$ and $m$), and notice that the number of items she gets is exactly $\ell$, no matter what the reports are.

We first prove inequality~\ref{inequality:NOM-worst-case} from the definition of $\NOM$, i.e. the worst-case guarantee. The $k$-th time agent $i$ gets to pick an item, there is an un-allocated item that she values at least $v_{i,i + (k-1)n}$, since only $(k-1)n + i - 1$ items have been allocated at that point. Therefore, the worst-case for agent $i$ under honest reporting is realized when all other agents rank the items in the same order as her, and her utility is exactly $\sum_{k=1}^{\ell} v_{i,i + (k-1)n}$. Now consider the case when all agents other than $i$ rank items in the same order as $i$ (i.e. $v_{i',j} \geq v_{i',j+1}$ for all $i' \in \agents$), and $i$ reports some vector $\bids_i$. Let $j_1, \dots, j_{\ell}$ be the items she receives, in the order she picked them. First, without loss of generality, $i$ picks these items in decreasing order with respect to her true valuation, i.e. $v_{i,j_{k}} \geq v_{i, j_{k+1}}$:
if she picks an item $k$ at round $t$, and at a later round $t'$ she could pick an item $k'$ with $v_{i,k'} > v_{i,k}$, then by the choice of valuation for agents other than $i$, picking $k'$ in round $t$ and $k$ in round $t'$ is also possible.
Second, $v_{i,j_k} \leq v_{i,i + (k-1)n}$: when $i$ picks an item for the $k$-th time, $(i-1) + (k-1)(n-1)$ of her top $i + (k-1)n$ items have been picked by other agents, and $k-1$ of her top $i + (k-1)n$ items have been picked by herself (since items are picked in decreasing order with respect to the true valuation). Therefore $v_{i,i + (k-1)n}$ is the largest value $v_{i,j_k}$ can have; $i$'s utility is at most $\sum_{k=1}^{\ell} v_{i,i + (k-1)n}$.

Next, we prove inequality~\ref{inequality:NOM-best-case} (the best-case guarantee). Since agent $i$ gets exactly $\ell$ items, her utility is at most $\sum_{k=1}^\ell v_{i,k}$.
This utility is realized when she reports $\values_i$, and everyone else ranks the items in the opposite order, i.e. $v_{i',j} < v_{i',j+1}$ for all $i'$, and it cannot be improved upon.
\end{proof}

\noindent The following corollary (in direct contrast to the results of~\cite{amanatidis2017truthful} for truthful mechanisms) follows immediately.

\begin{corollary}
There exists a deterministic, not obviously manipulable mechanism that achieves envy-freeness up to one item.
\end{corollary}


In fact, we can strengthen this result by proving there is a ``best-of-both-worlds'' mechanism that is not obviously manipulable. Specifically, we will show that the PS-Lottery algorithm of \cite{aziz2020simultaneously}, which outputs randomized allocations that are ex-ante EF and ex-post $\EFone$, is not obviously manipulable in expectation. 

The PS-Lottery algorithm is based on the well-known probabilistic serial algorithm, which outputs fractional allocations that are envy-free.  On a high level, the PS-Lottery algorithm uses Birkhoff's algorithm\footnote{Recall that, Birkhoff's algorithm, given a square bistochastic matrix, decomposes it into a convex combination (or a lottery) over permutation matrices.} to implement the fractional allocation output by probabilistic serial as a randomized allocation (a lottery) over a set of $\EFone$ allocations. For the sake of completeness, the PS-Lottery algorithm is formally described in Appendix~\ref{appendix:PS-Lottery}. 

We begin by proving a lemma that highlights a connection between not obvious manipulability and randomized mechanisms that output ex-ante proportional allocations. We note that a similar observation is made in~\cite{ortega2019obvious} in the context of cake cutting.

\begin{lemma}
\label{lemma:ex-ante-EF-and-NOM}
Inequality~\ref{inequality:NOM-worst-case} (the worst-case guarantee) is satisfied for every randomized mechanism $\mathcal{R}$ that outputs ex-ante proportional allocations.
\end{lemma}
\begin{proof}
Let $\mathcal{R}$ be a randomized mechanism that outputs ex-ante proportional allocations. Consider an agent $i \in \agents$ with true valuation $\values_i$.

Suppose agent $i$ reports her true valuation $\values_i$. Since the mechanism outputs ex-ante proportional allocations, for every possible reports of other agents, $\bids_{-i}$, the expected fractional allocation output by $\mathcal{R}$, $\allocy = (\allocy_1, \allocy_2, \ldots, \allocy_n)$ will be such that $\values_i(\allocy_i) \geq \frac{1}{n} \values_i(\items)$. As a consequence, when agent $i$ reports $\values_i$, her worst-case expected utility is at least $\frac{1}{n} \values_i(\items)$. 

Next, we show that when every agent $j \in \agents \setminus \{i\}$ reports her valuation to be $\values_i$ (the true valuation of agent $i$), then the worst-case expected utility of agent $i$, as per her true valuation and irrespective of her report, will be at most $\frac{1}{n} \values_i(\items)$, which implies the lemma. Consider the case when every agent $j \in \agents \setminus \{i\}$ reports her valuation to be $\bids_j = \values_i$ and agent $i$ reports valuation $\bids_i$. We know that the expected fractional allocation $\allocy' = (\allocy'_1, \allocy'_2, \ldots, \allocy'_n)$ returned by $\mathcal{R}$ will be proportional, i.e., for each agent $j \in \agents \setminus \{i\}$ we have $\bids_j(\allocy'_j) = \values_i(\allocy'_j) \geq \frac{1}{n} \values_i(\items)$. Adding up over all $j \in \agents \setminus \{i\}$ we have that $\sum_{j \in \agents \setminus \{i\}} \values_i(\allocy'_j) = \values_i(\items) - \values_i(\allocy'_i) \geq \frac{n-1}{n} \values_i (\items)$. On rearranging we get the required inequality $\values_i(\allocy'_i) \leq \frac{1}{n} \values_i(\items)$.
\end{proof}

Using Lemma~\ref{lemma:ex-ante-EF-and-NOM} we can prove the following theorem.

\begin{theorem}\label{best-of-both-worlds-and-NOM}
The PS-Lottery algorithm of~\cite{aziz2020simultaneously}, which is ex-ante envy-free and ex-post $\EFone$, is not obviously manipulable in expectation.
\end{theorem}


\begin{proof}
Inequality~\ref{inequality:NOM-worst-case} is implied by Lemma~\ref{lemma:ex-ante-EF-and-NOM}, since the PS-Lottery algorithm is ex-ante envy-free (see Appendix \ref{appendix:PS-Lottery}), and therefore ex-ante proportional.
It remains to prove Inequality~\ref{inequality:NOM-best-case}.

First, it holds that the expected fractional allocation returned by the PS-Lottery algorithm, irrespective of the agents' reports, is such that each agent receives $\frac{m}{n}$ (fractional) items in total; see Property \ref{PS-Lottery-property-2} in Appendix \ref{appendix:PS-Lottery}. Second, the best-case for an agent $i \in \agents$ who, without loss of generality, values items in the order $v_{i,1} \geq v_{i,2} \ldots \geq v_{i,m}$, and reports honestly, occurs when the reported valuation of other agents induce an opposite preference order on items, i.e., for each agent $j \in \agents \setminus \{i\}$ we have $b_{j,m} \geq b_{j,m-1} \ldots \geq b_{j,1}$. In this case, agent $i$ would receive items $1, 2 ,\dots, \lfloor \frac{m}{n} \rfloor$ in their entirety and a $\frac{m}{n} - \lfloor \frac{m}{n} \rfloor$ fraction of item $\lfloor \frac{m}{n} \rfloor + 1$. This allocation results in the maximum possible expected utility that agent $i$ can get subject to the constraint that she gets $\frac{m}{n}$ fraction of items, and therefore it cannot be improved upon, no matter what her report is.
\end{proof}

\section{Efficient Mechanisms}\label{sec: efficient}
We proceed to study whether natural efficiency notions are compatible with non obvious manipulability. 
We consider the three most popular notions of efficiency: utilitarian social welfare (sum of agents' utilities), egalitarian social welfare (the minimum utility of any agent), and Nash social welfare (geometric mean of agents' utilities). Note that for truthfulness, the only \emph{Pareto} efficient and truthful algorithm is a dictatorship, which immediately implies that one cannot truthfully achieve any non-trivial approximations with respect to any of these notions. 

\subsection{Utilitarian Social Welfare}\label{subsubsec: SW}

We start by considering the utilitarian social welfare maximizing algorithm, i.e. the algorithm which allocates each item to the agent that values it the most. When the winner for an item is not unique the algorithm needs to break ties; the choice of the tie-breaking rule will be crucial for our positive result in this section.

In the context of fair division it is standard to assume that agents' values are normalized when analyzing utilitarian social welfare. Specifically, the most common assumption is that the agents' values add up to $1$; see~\cite{aziz2020justifications} for a number of justifications for this assumption.
So, for the remainder of Section~\ref{subsubsec: SW} we will assume that $\sum_{j \in \items} v_{i,j} = 1$ for all $i \in \agents$.\footnote{Equivalently, we can define the utilitarian social welfare maximizing algorithm to be the aforementioned algorithm executed on transformed valuations where $\hat{v}_{i,j} = \frac{v_{i,j}}{\sum_{j \in \items} v_{i,j}}$ for all $i \in \agents$.} We note that without this assumption, the utilitarian social welfare maximizing algorithm is not obviously manipulable, since no matter what an agent reports, the best case for her is that she wins all items she values positively, and the worst case is that she loses all items, so no misreport can help increase either the worst-case or best-case utility.

\begin{theorem}
\label{theorem:NOM+SWM+2agents}
Every (randomized or deterministic) mechanism for $n=2$ agents that always outputs utilitarian social welfare maximizing allocations is obviously manipulable.
\end{theorem}

\begin{proof}
We prove the theorem for the case of two items; the proof can be easily generalized to hold for any number of items.
Let $\mathcal{M}$ be the utilitarian social welfare maximizing algorithm, coupled with any tie-breaking rule. 
Assume that both agents report a value of $1$ for the first item, and zero for the other item, i.e. $b_{i,1} = 1$ and $b_{i,2} = 0$ for $i \in \{ 1, 2 \}$.
Given these reports, there must be an agent who gets item $1$ with probability at least $\frac{1}{2}$ (if $\mathcal{M}$ is deterministic then this probability will be exactly $1$); assume that this is agent $1$, without loss of generality.

Now, consider the case when the true valuation of agent $1$ is $\values_1 = (\frac{2}{3}+\epsilon, \frac{1}{3} - \epsilon)$, for some small $\epsilon > 0$. Additionally, suppose that agent $1$ reports her true valuation (i.e., $\bids_1 = \values_1$) and agent $2$'s reported value for item $1$, $b_{2,1} > \frac{2}{3}+\epsilon$. In this case, the utilitarian social welfare maximizing allocation gives the first item to agent $2$. Therefore, the worst-case utility of agent $1$ when she reports her true valuation is at most $\frac{1}{3} - \epsilon$. Next, consider the dishonest report $\bids_1 = (b_{1,1}, b_{1,2}) = (1,0)$. Given this, if $b_{2,1} < 1$, then agent $1$ gets the first item for a utility of $\frac{2}{3}+\epsilon$. Otherwise if $b_{2,1} = 1$, then agent $1$ gets the first item with probability at least $1/2$ (this follows from our choice of agent $1$), thus, her expected utility is at least $\frac{1}{2}(\frac{2}{3} + \epsilon) =  \frac{1}{3} + \frac{\epsilon}{2}$. In either case, her utility is strictly larger than $\frac{1}{3} - \epsilon$, her worst-case utility under honest reporting.
Therefore, $\mathcal{M}$ is obviously manipulable.
\end{proof}

Surprisingly, this impossibility result does not hold for more than two agents. 

\begin{theorem}
\label{theorem:NOM+SWM+nagents}
The utilitarian social welfare maximizing algorithm, coupled with an appropriate tie-breaking rule, is not obviously manipulable for $n \geq 3$ agents.
\end{theorem}

\begin{proof}
First, we describe our algorithm and tie-breaking rule. The algorithm, given reports $\bids_1, \dots, \bids_n$, allocates each item to an agent with the largest reported value for that item. In case of a tie, the item is allocated to the agent with the smallest index (i.e. $1 \succ 2 \succ 3 \dots$), \emph{except} if the tie is exactly between agents $1$ and $n$, in which case the item should go to agent $n$. Towards showing that this algorithm is NOM, consider an arbitrary agent $i$ with true valuation $\values_i$.

Towards proving inequality~\ref{inequality:NOM-worst-case}, the worst-case guarantee, we have that if $b_{i,j} < 1$ for all $j \in \items$, it is easy to see that there exists a choice for $\values_{-i}$ such that agent $i$ does not have the highest value for any item (e.g., some agent $k$ could have a value of $1$ for $i$'s favorite item, and a different agent $k'$ can out-bid $i$ in all remaining items). If $b_{i,j} = 1$ for some $j \in \items$, then, again, $i$ can again end up with no items, since by the choice of our tie-breaking rule, all agents can lose in the tie-breaking (agents other than $1$ lose to some smaller index agent, while agent $1$ loses to agent $n$). Therefore, no matter what agent $i$ reports, in the worst-case she gets a utility of zero.

Next we prove inequality~\ref{inequality:NOM-best-case}, the best-case guarantee. By the normalization assumption and tie-breaking rule, agent $i$ cannot out-bid every single agent on every single item, no matter what $\bids_i$ and $\values_{-i}$ are. Therefore, the best-case outcome for agent $i$ is that she is allocated all items except her least favorite one. Finally, notice that this outcome can be realized when reporting honestly, if $\values_{-i}$ is such that $v_{i',j} = 1$, for all $i' \neq i$, where $j$ is the item that $i$ values the least (i.e. $j \in argmin_{k \in \items} v_{i,k})$.
\end{proof}


\subsection{Egalitarian Social Welfare}

In this section we will show that any mechanism --- randomized or deterministic --- that maximizes  egalitarian social welfare (the minimum utility among agents) is obviously manipulable. 
This result rules out the existence of NOM mechanisms that output leximin allocations\footnote{An allocation is leximin iff it maximizes the lowest utility, and subject to that the second lowest utility and so on.}, since leximin allocations are, by definition, utilitarian social welfare maximizing. 

Similar to utilitarian social welfare, when discussing egalitarian social welfare, we 
will assume that the valuations of agents are normalized: for every agent, the combined value for the set of all items is $1$; see~\cite{aziz2016egalitarianism} for a thorough discussion. 


For the case of $n=2$ agents, there are no NOM mechanisms that maximize egalitarian social welfare for that case; this follows from Theorem~\ref{theorem:NOM-negative-result-combo} in Section~\ref{sec: best of both}. The following theorem establishes that this impossibility continues to hold for the case of $n \geq 3$ agents.

\begin{theorem}\label{thm: NOM + egalitarian + n agents}
Every (randomized or deterministic) mechanism that always outputs an allocation which maximizes the egalitarian social welfare is obviously manipulable, even for $n=3$ agents and $m = 4$ items.
\end{theorem}

We defer the proof of Theorem~\ref{thm: NOM + egalitarian + n agents} to Appendix~\ref{app:missing}.

\subsection{Nash Social Welfare}

The next natural maximization objective, and arguably the most popular one in the context of fair division, is the Nash social welfare. Here, we show that --- unlike utilitarian social welfare maximization --- there does not exist a not obviously manipulable mechanism that always outputs allocations that maximize the Nash social welfare, for any number of agents. The case of $n=2$ agents follows from Theorem~\ref{theorem:NOM-negative-result-combo} in Section~\ref{sec: best of both}. In the following theorem, we show that the impossibility extends to $n \geq 3$ agents. 

\begin{theorem}\label{thm: MNW impossibility}
Every (randomized or deterministic) mechanism $\mathcal{M}$ that always outputs an allocation that maximizes the Nash social welfare is obviously manipulable, even for $n=3$ additive agents and $m=4$ items.
\end{theorem}
\begin{proof}
We prove the theorem for the case of $n=3$ agents and $m=4$ items; we describe how our arguments can be adjusted to work for $n > 3$ agents and any number of items at the end of this proof. Let $\mathcal{M}$ be a (possibly randomized) mechanism that \emph{always} outputs a Nash social welfare maximizing allocation.

Let the true valuation of agent $1$ be $\values_1 = (3.9, 3, 2, 0.9)$. The subsequent proof has two parts: first, we will show that if agent $1$ reports her true valuation, then the worst-case utility is exactly $2$, and second, if agent $1$ reports $\bids_1 = (2, 2, 1, 1)$, then the worst-case would be strictly more than $2$; the theorem follows. \\

\noindent
\emph{Worst-case utility when reporting honestly:} Consider the case where agent $1$ reports her true valuation, i.e., the report $\bids_1 = \values_1$. We know that mechanism $\mathcal{M}$ must output Nash social welfare maximizing allocations, and such allocations are necessarily EF1 \cite{caragiannis2019unreasonable}; randomized $\mathcal{M}$ will output ex-post EF1 allocations. Consequently, agent $1$ must be allocated at least one item, since otherwise some other agent will receive at least two items, and agent $1$ will envy that agent even upon the removal of any one item. Furthermore, if agent $1$ is allocated only item $4$ the overall allocation can't be EF1, since agent $1$ will envy (even upon the removal of any item) the agent that gets two of the first three items ($v_{1,4}$ is smaller than all other values). Thus, in the worst-case, agent $1$ will get a bundle whose value is at least $2$, her value for item $3$. Next, we show that her worst-case utility is exactly equal to $2$.

Let the reported valuations of agent $2$ be $\bids_2 = (0, 1, 0, 0)$ and of agent $3$ be $\bids_3 = (2, 0, 0, 1)$. Given this, the Nash welfare maximizing allocation is unique, and this allocation is such that agent $1$ gets only item $3$, for a total utility of $2$. To see that this allocation is unique, first notice that item $2$ must go to agent $2$ (otherwise agent $2$'s utility, and therefore Nash social welfare, will be zero). Additionally, item $3$ must go to agent $1$ since she is the only agent with a non-zero value for it. The remaining items ($1$ and $4$), which are valued only by agents $1$ and $3$, must be allocated in a way that the Nash social welfare of the resultant allocation is maximized. A simple case analysis shows that the unique allocation that maximizes Nash social welfare gives both items to agent $3$. A consequence of having a unique Nash social welfare maximizing allocation is that mechanism $\mathcal{M}$ must output it, irrespective of whether $\mathcal{M}$ is randomized or deterministic.\\




\noindent
\emph{Worst-case utility when misreporting:} Consider the case where agent $1$ misreports her valuation as $\bids_1 = (2, 2, 1, 1)$. Given this report, the mechanism $\mathcal{M}$ must allocate at least one item to agent $1$; otherwise its allocation would not be EF1 (ex-post EF1 for if $\mathcal{M}$ is randomized). We will show that agent $1$'s allocation cannot be only item $3$ nor only item $4$ (i.e., a bundle of value $1$ with respect to the reported values), irrespective of the reports of agents $2$ and agent $3$. Specifically, agent $1$ must be allocated either $(i)$ at least two items, $(ii)$ only item $1$, or $(iii)$ only item $2$. In each of these cases, the value of the items allocated to agent $1$ (as per her true valuation) will be strictly more than $2$, which completes the proof. 

We proceed to show that agent $1$ cannot be allocated just item $3$ or just item $4$ in any Nash social welfare maximizing allocation. Towards a contradiction, assume that in a Nash social welfare maximizing allocation $A$, agent $1$ is allocated either only item $3$ or only item $4$. Since agent $1$ is allocated only one item, there must be some other agent who is allocated at least two items. We consider the following exhaustive cases based on the items allocated to this other agent.\\

\noindent
\emph{Case {\rm I}}: There is an agent $j \in \{2,3\}$ having both items $1$ and $2$. In this case, the allocation $A$ is not EF1 --- and hence not maximizing Nash social welfare --- since agent $1$ envies agent $j$ (with respect to the reported valuation $\bids_1$), even upon the removal of one item. \\

\noindent
\emph{Case {\rm II}}: There is an agent $j \in \{2,3\}$ who gets one item from the set of items $\{ 1, 2 \}$ and one item from $\{ 3, 4 \}$. Without loss of generality assume that agent $j$ is allocated items $2$ and $3$, and agent $1$ gets item $4$. We will show that in this case $A$ can never be a Nash social welfare maximizing allocation. 

Since allocation $A$ maximizes Nash social welfare, it must be that transferring item $2$ from agent $j$ to agent $1$ does not increase the Nash social welfare, i.e., the following inequalities must hold:
\begin{align}
    b_{1,4} \cdot (b_{j, 2} + b_{j,3}) & \geq b_{j,2} (b_{1,3} + b_{1,4}) \notag \\
    1 \cdot (b_{j, 2} + b_{j,3}) & \geq b_{j,2} \cdot (1 + 1) \tag{substituting $b_{1,3} = b_{1,4} = 1$}\\
    b_{j,3} & \geq b_{j,2} \label{inequality:MNW-impossibility-1}
\end{align}
Similarly, transferring item $3$ from agent $j$ to agent $1$ must also not increase the Nash social welfare:
\begin{align}
    b_{1,4} \cdot (b_{j, 2} + b_{j,3}) & \geq b_{j,3} \cdot (b_{1,2} + b_{1,4}) \notag \\
    1 \cdot (b_{j, 2} + b_{j,3}) & \geq b_{j,3} \cdot (2 + 1) \tag{substituting $b_{1,2} = 2$ and $b_{1,4} = 1$}\\
    b_{j,2} & \geq 2 b_{j,3} \label{inequality:MNW-impossibility-2}
\end{align}
Combining inequalities (\ref{inequality:MNW-impossibility-1}) and (\ref{inequality:MNW-impossibility-2}) we have $b_{j,3} \geq 2 b_{j,3}$, which can be true only if $b_{j,3} = 0$. However, since $A$ allocates item $3$ to agent $j$, and $b_{j,3} = 0$, but $b_{1,3} > 0$, $A$ is not Nash social welfare maximizing (in fact, not even Pareto efficient).


The same argument can be extended to instances having $n>3$ agents and $m = n+1$ items by considering the case wherein agent $1$'s true valuation, $\values_1 = (3.9, 3, 3, \ldots, 3, 2, 0.9)$. Here, the misreported valuation which leads to an improvement in her (expected) worst-case utility is $\bids_1 = (2, 2, 2, \ldots, 2, 1, 1)$.
\end{proof}

\section{Fair and Efficient Mechanisms}\label{sec:main}
Maximizing the Nash social welfare results in allocations that are fair in addition to being economically efficient ~\cite{caragiannis2019unreasonable}. However, as we established in Theorem~\ref{thm: MNW impossibility}, mechanisms that output MNW allocations are obviously manipulable for any number of agents. 

In this section we state our main result, where we show that there is a deterministic $\NOM$ mechanism that is fair and economically efficient for additive valuations. This result is established by showing a black-box reduction from the problem of designing $\NOM$ mechanisms that output $\EFone$ allocations to the problem of designing algorithms that output $\EFone$ allocations. Additionally, this black-box reduction preserves Pareto efficiency guarantees.

The following theorem formally states our main result.

\begin{theorem}\label{theorem:NOM+PO+EF1}
For additive valuations, there exists a black-box reduction from the problem of designing a not obviously manipulable and EF1 mechanism, to designing an algorithm that computes clean, non-wasteful and EF1 allocations. Additionally, the reduction preserves Pareto efficiency guarantees.
Specifically, given an algorithm $\mathcal{M}^*$ that computes clean, non-wasteful and EF1 allocations, we can construct a deterministic not obviously manipulable mechanism that outputs EF1 partial allocations by using Mechanism \ref{mechanism:NOM+PO+EF1}. Moreover, if $\mathcal{M}^*$ is $\alpha$-Pareto efficient (resp. $\alpha$-fractionally Pareto efficient) then Mechanism~\ref{mechanism:NOM+PO+EF1} always outputs $\alpha$-Pareto efficient (resp. $\alpha$-fractionally Pareto efficient) partial allocations.
\end{theorem}

By combining Theorem~\ref{theorem:NOM+PO+EF1} with known algorithmic results we can get NOM mechanisms, with the same fairness and efficiency guarantees. Specifically,~\cite{garg2021fair} give a pseudo-polynomial time algorithm that computes fractionally Pareto efficient and $\EFone$ allocations. Fractional Pareto efficiency implies non-wastefulness, and without loss of generality we can assume that this algorithm outputs clean allocations.\footnote{Starting from an fPO and EF1 allocation $A = (A_1, A_2, \ldots, A_n)$, if we clean each bundle $A_i$, then the resultant partial allocation is fPO and EF1 as well.} Hence, by applying Theorem~\ref{theorem:NOM+PO+EF1} we get the following application.

\begin{corollary}\label{application: additive}(via ~\cite{garg2021fair}).
For agents with additive valuations, there exists a fractionally Pareto efficient, $\EFone$, and $\NOM$ mechanism, that runs in pseudo-polynomial time.
\end{corollary}

\subsection{The Reduction}
Our reduction, Mechanism~\ref{mechanism:NOM+PO+EF1}, takes as input reported valuations $\bids = (\bids_1, \dots, \bids_n)$, and black-box access to a deterministic algorithm $\mathcal{M}^*$.\footnote{Algorithm $\mathcal{M}^*$ could possibly be computationally inefficient.} 
Our reduction requires the algorithm $\mathcal{M}^*$ to always output non-wasteful, clean EF1 allocations for every possible input valuation functions.

The reduction is based on two key ideas, first, through a careful construction of cases we ensure that if an agent $i\in \agents$ reports valuation $\bids_i$, then the set of allocations that can be returned by our reduction, Mechanism~\ref{mechanism:NOM+PO+EF1}, for every possible $\bids_{-i}$ can be precisely characterized (Lemma \ref{lemma:every-EF1-realized}). Second, we prove a structural result (Lemma \ref{lemma:least-valued-EF1}) concerning this set of possible output allocation. This structural result plays a central role in establishing that Mechanism~\ref{mechanism:NOM+PO+EF1} is not obviously manipulable. Further, we prove that by construction our reduction always outputs fPO + EF1 allocations.

\begin{algorithm}[ht]
\caption{Black-box reduction}
\label{mechanism:NOM+PO+EF1}
\hspace*{\algorithmicindent}\textbf{Input:} Reported valuation functions of agents $\bids$. Black-box access to an algorithm $\mathcal{M}^*$. \\
\hspace*{\algorithmicindent}\textbf{Output:} A partial integral allocation $\Alloc = (A_1, \ldots, A_n)$ 
  \begin{algorithmic}[1]
  		\STATE \textbf{Set} $D_i \gets \{ g \in \items \ | \ b_{i,g} > 0 \}$. 
  		\STATE \textbf{Set} $\widehat{D}_i \gets \items \setminus \cup_{j\neq i} D_j$ for each $i \in \agents$.
		\STATE \textbf{Set} $R_i \gets 1$ if subsets $\{D_j\}^n_{j=1} \setminus \{D_i\}$ are pairwise disjoint; $0$ otherwise.
		\tikzmark{top1}
  		\IF{the subsets $\{D_j\}_{j=1}^n$ are pairwise disjoint 	}
  			\STATE $\Alloc^* \gets (D_1, \ldots, D_n)$
  		\tikzmark{bottom1}
  		\tikzmark{top2}
  		\ELSIF{$\exists i \in \agents$ such that $R_i = 1$ and $R_j = 0$ for every $j \in \agents \setminus \{i\}$} \label{case2-begin}
  			\IF{the allocation $(D_1, D_2, \ldots, \widehat{D}_i, \ldots, D_n)$ is $\EFone$}
  				\STATE $\Alloc^* \gets (D_1, D_2, \ldots, \widehat{D}_i, \ldots, D_n)$ \label{case2-return1}
  			\ELSE
  				\STATE $\Alloc^* \gets \mathcal{M}^*(\bids_1, \bids_2, \ldots, \bids_n)$
  			\ENDIF \label{case2-end}
  		\tikzmark{bottom2}
  		\tikzmark{top3}
  		\ELSIF{$\exists i,j \in \agents$ such that $i<j$ and $R_i = R_j = 1$, and $R_k =0$ for $k \in [n] \setminus \{i,j\}$}~~~~~~\tikzmark{right} \label{case3-begin}
  			\IF{$\bids_i(D_i \cap D_j) < \bids_j(D_i \cap D_j)$} \label{case3-tie-breaking:begin}
				\IF{$(D_1, \ldots, \widehat{D}_i, \ldots, D_j, \ldots, D_n)$ is $\EFone$} \label{case3-subcase1-if} 
					\STATE $\Alloc^* \gets (D_1, \ldots, \widehat{D}_i, \ldots, D_j, \ldots, D_n)$ \label{case3-subcase1-return}
				\ELSE
					\STATE $\Alloc^* \gets \mathcal{M}^*(\bids_1, \bids_2, \ldots, \bids_n)$
				\ENDIF  \label{case3-subcase1:end}		
  			\ELSE \label{case3-subcase2:begin}
				\IF{$(D_1, \ldots, D_i, \ldots, \widehat{D}_j, \ldots, D_n)$ is $\EFone$} 
					\STATE $\Alloc^* \gets (D_1, \ldots, D_i, \ldots, \widehat{D}_j, \ldots, D_n)$
				\ELSE
					\STATE $\Alloc^* \gets \mathcal{M}^*(\bids_1, \bids_2, \ldots, \bids_n)$ 
				\ENDIF 
  			\ENDIF \label{case3-tie-breaking:end}
  			\tikzmark{bottom3}
  		\tikzmark{top4}
  		\ELSE
  			\STATE $\Alloc^* \gets \mathcal{M}^*(\bids_1, \bids_2, \ldots, \bids_n)$
  		\ENDIF
  		\tikzmark{bottom4}
  		\STATE For all $i \in \agents$, iteratively remove goods $g \in \Alloc^*_i$ such that $\bids_i( \Alloc^*_i \setminus \{ g \} ) = \bids_i( \Alloc^*_i )$. \label{code:removal}
  		\STATE \textbf{return} $\Alloc^*$
		\end{algorithmic}
		\AddNote{top1}{bottom1}{right}{\emph{Case $\rm{I}$}}
		\AddNote{top2}{bottom2}{right}{\emph{Case $\rm{II}$}}
		\AddNote{top3}{bottom3}{right}{\emph{Case $\rm{III}$}}
		\AddNote{top4}{bottom4}{right}{\emph{Case $\rm{IV}$}}
\end{algorithm}


We begin by defining notations required to describe our reduction. For each agent $i \in \agents$, let $D_i$ be the set of goods that have strictly positive value for $i$ i.e. $D_i \coloneqq \{ g \in \items \ | \ b_{i,g} > 0 \}$. 
Let $\widehat{D}_i \coloneqq \items \setminus \cup_{j\neq i} D_j$ be the goods remaining after removing all goods desired by agents other than $i$. Finally, let $R_i$, for each agent $i \in \agents$, be the indicator for the event that the subsets $( \{D_j\}_{j \in \agents}) \setminus \{D_i\}$ are pairwise disjoint, i.e. $R_i = 1$ iff $( \{D_j\}_{j \in \agents}) \setminus \{D_i\}$ are pairwise disjoint, and $R_i = 0$ otherwise. Note that the $R_i$s can be easily computed in polynomial time.

Mechanism~\ref{mechanism:NOM+PO+EF1} sequentially considers four cases based on the sets $\{D_j\}^n_{j=1}$ and the values $\{R_j\}^n_{j=1}$, to find a temporary assignment $\Alloc^*$, i.e., a (partial) integral allocation.

\begin{description}
\item[\emph{Case \rm{I}}:] The sets $\{D_j\}^n_{j=1}$ are pairwise disjoint (equivalently, $R_i = 1$ for all $i \in \agents$). In this case, the temporary assignment allocates the bundle $D_j$ to agent $j$ for each agent $j \in \agents$.
\item[\emph{Case \rm{II}}:] $R_i = 1$ for exactly one agent $i \in \agents$. This can occur if $D_i$ intersects with two or more $D_j$s (and these are the only intersections among pair of subsets in $\{D_k\}_{k=1}^n$). In this case, if allocating the bundle $\widehat{D}_i$ to agent $i$, and the bundle $D_j$ to each agent $j \in \agents$, for each $j \neq i$ results in an $\EFone$ allocation, then the temporary assignment is set to this allocation. Otherwise, if this allocation is not $\EFone$, then the temporary assignment is the allocation returned by the black-box algorithm $\mathcal{M}^*$ for the given valuation profile.
\item[\emph{Case \rm{III}}:] There are exactly two agents $i,j \in \agents$ such that $R_i = R_j = 1$. Note that the only way this is possible is if $D_i$, $D_j$ intersect each other and any other pair of subsets $D_k, D_l$ where $\{k,l\} \neq \{i,j\}$, are disjoint. In this case, Mechanism~\ref{mechanism:NOM+PO+EF1} considers whether the set of goods $D_i \cap D_j$ are valued more by agent $i$ or agent $j$; each of these two subcases are similar to \emph{Case \rm{II}}, see Lines \ref{case3-tie-breaking:begin}-\ref{case3-tie-breaking:end}.
\item[\emph{Case \rm{IV}}:] None of the previous cases holds (equivalently, $R_i = 0$ for all $i \in \agents$). In this case, the temporary assignment is simply the allocation returned by the black-box algorithm $\mathcal{M}^*$ for the given valuation profile.
\end{description}

The last step of Mechanism~\ref{mechanism:NOM+PO+EF1} (Line~\ref{code:removal}) is to make the bundle allocated to each agent $i$ in the temporary assignment \emph{clean}, i.e. remove items from her bundle that she does not value. This step is necessary for one of our technical lemmas (specifically, for characterizing the set of all allocations that are possible outputs of our mechanism). Additionally, note that this step does not affect efficiency or envy-freeness up to one item (i.e. if the allocation satisfied any of these notions before this step, it continues to do so).

\textbf{The case of strictly positive valuations.} It is interesting to notice that when valuations are strictly positive, i.e. $v_{i,j} > 0$ for all $i,j$, Mechanism~\ref{mechanism:NOM+PO+EF1} simply outputs the allocation of $\mathcal{M}^*$ (Case \rm{IV}). However, this does \emph{not} imply that every PO and EF1 algorithm is NOM: in calculating worst-case outcomes, agents consider the possibility that others have zero valuations. This fact is critical in our proof of correctness (see the proof of Lemma~\ref{lemma:every-EF1-realized}).

\paragraph{Establishing the main result}

Towards establishing our main result, Theorem \ref{theorem:NOM+PO+EF1}, we begin by proving the following supporting lemmas. We start by showing that Mechanism~\ref{mechanism:NOM+PO+EF1} preserves the efficiency properties of $\mathcal{M}^*$. 
We will then proceed to show that Mechanism~\ref{mechanism:NOM+PO+EF1} also preserves EF1 if $\mathcal{M}^*$ is EF1, and, in fact, Mechanism~\ref{mechanism:NOM+PO+EF1} is also NOM when $\mathcal{M}^*$ is EF1.

\begin{lemma}
\label{lemma:mech is PO}
If $\mathcal{M}^*$ is $\alpha$-Pareto efficient (resp. $\alpha$-fractionally Pareto efficient) then Mechanism~\ref{mechanism:NOM+PO+EF1} always outputs an $\alpha$-Pareto efficient (resp. $\alpha$-fractionally Pareto efficient) partial allocation.
\end{lemma}

\begin{proof}
First, note that Line~\ref{code:removal} (making the temporary assignment clean) does not affect the utility of any agent $i$ for her own bundle.
Therefore if the allocation before this step was $\alpha$-PO (respectively $\alpha$-fPO) then after Line~\ref{code:removal} these efficiency guarantees will continue to hold. Considering all possible cases for Mechanism~\ref{mechanism:NOM+PO+EF1}, if all agents $j \in \agents$ receive the bundle $D_j$ (as in case \rm{I}), the allocation is fractionally Pareto efficient. If all agents $j$ receive the bundle $D_j$, with the exception of a single agent $i$ that receives $\widehat{D}_i$, then all agents except $i$ have the maximum possible utility, i.e. their utility for all items $\items$.
Furthermore, it is impossible to improve the utility of $i$ without decreasing the utility of some other agent $j$; this follows from the definition of $\widehat{D}_i$ and $D_j$. Therefore, this allocation is fractionally Pareto efficient as well. This allocation is considered in cases \rm{II} and \rm{III} (if the black-box algorithm $\mathcal{M}^*$ is not called), and it is the output of Mechanism~\ref{mechanism:NOM+PO+EF1} if it is also $\EFone$. In the remaining cases (if the aforementioned allocation is not $\EFone$, or we are in case \rm{IV}) Mechanism~\ref{mechanism:NOM+PO+EF1} returns the allocation computed by the black-box algorithm $\mathcal{M}^*$ (modulo Line~\ref{code:removal}), which is an $\alpha$-PO (respectively $\alpha$-fPO) partial allocation by definition.
\end{proof}

\begin{lemma}
\label{lemma:mech is EF1}
If $\mathcal{M}^*$ outputs non-wasteful, clean, and $\EFone$ allocations then Mechanism~\ref{mechanism:NOM+PO+EF1} always outputs non-wasteful, clean, and $\EFone$ partial allocations.
\end{lemma}

\begin{proof}
Line~\ref{code:removal} does not affect the utility of any agent for her own bundle, and since valuations are additive, removal of items could only decrease her utility for the bundle of a different agent. Therefore if the allocation before this step was $\EFone$, it remains $\EFone$ afterwards. This step also ensures that the allocation is clean. Next, considering all possible cases, we first have that if every agent $j \in \agents$ receives the bundle $D_j$ (as in case \rm{I}), the allocation is envy-free, and therefore $\EFone$. In cases \rm{II} and \rm{III}, Mechanism~\ref{mechanism:NOM+PO+EF1} either checks whether an allocation is $\EFone$ before outputting it, or calls $\mathcal{M}^*$. In case \rm{IV} $\mathcal{M}^*$ is called. Whenever $\mathcal{M}^*$ is called, then the output allocation is $\EFone$ by definition. 
Finally, note that Mechanism~\ref{mechanism:NOM+PO+EF1} returns non-wasteful allocations: it is easy to confirm that whenever Mechanism~\ref{mechanism:NOM+PO+EF1} doesn't call $\mathcal{M}^*$ its allocation is non-wasteful, while if $\mathcal{M}^*$ is called, the allocation of $\mathcal{M}^*$ the returned allocation is non-wasteful and clean, so it remains non-wasteful after Line~\ref{code:removal}.\footnote{Note that if the allocation of $\mathcal{M}^*$ was non-wasteful and not clean, then Line~\ref{code:removal} could have altered this allocation into a wasteful one.}

\end{proof}

The next lemma characterizes the set of allocations that Mechanism~\ref{mechanism:NOM+PO+EF1} could possibly return given the reported valuation of a particular agent. 
Specifically, we will show that every clean, non-wasteful and $\EFone$ allocation, that is consistent with the reported valuation, is a possible output of Mechanism \ref{mechanism:NOM+PO+EF1}. Before stating the lemma, we define some useful notation.
For any agent $i \in \agents$ and valuation function $\values_i$, let $\EFone(i,\values_i)$ be the set of (partial) allocations $\Alloc = (A_1, \ldots, A_n)$ that are clean, non-wasteful and envy-free up to one item with respect to agent $i$ when her valuation function is $\values_i$, i.e., 
$\EFone(i,\values_i) = \{ \Alloc \in \Pi_n(\items) \ | \ \forall g \in A_i ~~ v_{i,g} > 0, \text{ and } \forall j \in \agents \text{ with } A_j \neq \emptyset, \values_i(A_i) \geq \values_i(A_j \setminus \{g\}) \text{ for some }g \in A_j, \text{ and }  \values_i( \items \setminus \cup_{k \in \agents} A_k ) = 0  \}$.

\begin{lemma}
\label{lemma:every-EF1-realized}
Given any agent $i \in \agents$ and valuation function $\values_i$, for every allocation $\Alloc \in \EFone(i,\values_i)$ there exists a set of valuations $\values_{-i}$ such that Mechanism~\ref{mechanism:NOM+PO+EF1} on input  $(\values_i, \values_{-i})$ outputs allocation $\Alloc$.
\end{lemma}

\begin{proof}
Given the valuation vector $\values_i$ and an allocation $\Alloc = (A_1, A_2, \dots, A_n) \in \EFone(i,\values_i)$ we will construct valuations for the other agents, $\values_{-i}$, such that Mechanism~\ref{mechanism:NOM+PO+EF1} outputs $\Alloc$ on input  $(\values_i, \values_{-i})$.  The valuation of each agent $j \in (\agents \setminus \{i\})$ will be such that $v_{j,g} = 0$ for all $g \notin A_j$, and $v_{j,g} > 0$ for each good $g \in A_j$ (the precise value of $v_{j,g}$ will depend on how $D_i$ intersects with $A_j$). That is, for each agent $j \in (\agents \setminus \{i\})$, $D_j = A_j$. Therefore, by construction, the subsets $\{D_j\}^n_{j=1} \setminus \{D_i\}$ are pairwise disjoint, and therefore $R_i = 1$ (so, we are never in case $\rm{IV}$ of Mechanism~\ref{mechanism:NOM+PO+EF1}). Also, by construction, Line~\ref{code:removal} will not affect any bundle.
In the rest of the proof, we consider the three exhaustive cases, based on how many sets from $\Alloc_{-i} = (A_1, \dots, A_{i-1}, A_{i+1}, \dots, A_n)$ 
the set $D_i$ intersects.

The first case is when $D_i$ does not intersect any bundle from $\Alloc_{-i}$. Since $v_{i,g} = 0$ for every unallocated good $g$ (by the definition of $\EFone(i,\values_i)$), $D_i \subseteq A_i$, and therefore $R_j = 1$ for all $j \in \agents$. In this case we can set $v_{j,g}$ to be an arbitrary value for all $g \in A_j$, e.g. $v_{j,g} = 1$, for all agents $j \neq i$.
Therefore, Mechanism~\ref{mechanism:NOM+PO+EF1} considers $\emph{Case \rm{I}}$, and it allocates the set $D_j$ to every agent $j$. By construction, for all $j \neq i$, $D_j = A_j$. Additionally, for agent $i$, it must be that the bundle $A_i = D_i$, since $A_i$ is also clean (by Line~\ref{code:removal}). Therefore, the output of Mechanism~\ref{mechanism:NOM+PO+EF1} is precisely $\Alloc$. 

The second case is that $D_i$ intersects exactly one bundle from $\Alloc_{-i}$. Let $j^* \in (\agents \setminus \{i\})$ be the (unique) agent such that $A_{j^*} \cap D_i \neq \emptyset$. Additionally, assume that $i<j^*$; an almost identical argument works if $i > j^*$. 
For each agent $k \in \agents \setminus \{i,j^*\}$, we define $v_{k,g} = 1$ if good $g \in A_k$ and $v_{k,g} = 0$ otherwise. The valuations of agent $j^*$ are as follows
\[ 
v_{j^*,g} = \begin{cases} 
      2  \values_i (D_i \cap D_{j^*}) & g \in A_{j^*} \cap D_i \\
      1 & g \in A_{j^*} \setminus D_i \\
      0 & otherwise 
   \end{cases}
\]
Note that, the set of desired goods $D_k = A_k$ for all agents $k \in (\agents \setminus \{i\})$, and $\widehat{D}_i \supseteq A_i$. Furthermore, since $\Alloc$ is non-wasteful (by the definition of $\EFone(i,\values_i)$), goods in $\widehat{D}_i \setminus A_i$ must have zero value for $i$. Therefore, $\values_i(\widehat{D}_i) = \values_i(A_i)$. The construction satisfies that the bundles $\{D_k\}_{k \in \agents \setminus \{i\}}$ are pairwise disjoint, and the bundle $D_i$ only intersects $D_{j^*} = A_{j^*}$. Therefore, $R_i = R_{j^*} = 1$ and $R_k = 0$ for every other agent $k$. Therefore, Mechanism~\ref{mechanism:NOM+PO+EF1}, given valuations $(\values_i, \values_{-i})$, considers $\emph{Case \rm{III}}$ (Lines \ref{case3-begin}-\ref{case3-tie-breaking:end}) to compute the final allocation. In $\emph{Case \rm{III}}$, the mechanism first checks whether $\values_i (D_i \cap D_{j^*}) < \values_j (D_i \cap D_{j^*})$ holds, which is true for our construction since $v_{j^*,g} = 2 \values_i (D_i \cap D_{j^*})$ for all $g$ in $A_{j^*} \cap D_i = D_{j^*} \cap D_i$. Afterwards, the mechanism checks if the allocation $(D_1, \ldots, \widehat{D}_i, \ldots, D_{j^*}, \ldots, D_n)$ is $\EFone$ (Line \ref{case3-subcase1-if}). Since $\values_i(\widehat{D}_i) = \values_i(A_i)$, and $\Alloc \in \EFone(i,\values_i)$, this allocation is indeed $\EFone$ for agent $i$, and since all other agents are envy-free, Mechanism~\ref{mechanism:NOM+PO+EF1} sets it as the temporary assignment. After removing all zero valued items from $\widehat{D}_i$ (Line~\ref{code:removal}) we are left with the bundle $A_i$, hence, the final allocation is exactly $\Alloc$. 

The third and final case is that $D_i$ intersects more than one bundle from $\Alloc_{-i}$. Here, the valuation vector of every agent $j \in \agents \setminus \{i\}$ is such that $v_{j,g} = 1$ for $g \in A_j$ and $v_{j,g} = 0$ otherwise. By construction, the set of desired goods $D_j = A_j$ for every agent $j \in \agents \setminus \{i\}$, $\widehat{D}_i \supseteq A_i$, and since $\Alloc$ is non-wasteful, we have $\values_i(\widehat{D}_i) = \values_i(A_i)$. $D_i$ intersecting more than one bundle from $\Alloc_{-i}$, is equivalent to $D_i$ intersecting more than one subsets from $\{D_j\}_{j \neq i}$, and therefore $R_j = 0$ for all $j \in \agents \setminus \{i\}$. Thus, given the valuation profile $(\values_i, \values_{-i})$, Mechanism~\ref{mechanism:NOM+PO+EF1} considers $\emph{Case \rm{II}}$ (Lines \ref{case2-begin}-\ref{case2-end}), where it first checks whether the allocation $(D_1, D_2, \ldots, \widehat{D}_i, \ldots, D_n)$ is $\EFone$. This allocation has the same value as $\Alloc$ for agent $i$ (and all other agents have exactly the same allocation), so agent $i$ has an envy of at most one item by definition, and all other agents have no envy at all. Therefore, Mechanism~\ref{mechanism:NOM+PO+EF1} sets this as the temporary assignment. After removing all zero valued items from $\widehat{D}_i$ (Line~\ref{code:removal}) the final allocation is exactly $\Alloc$. This concludes the proof of Lemma~\ref{lemma:every-EF1-realized}.
\end{proof}

Lemma~\ref{lemma:every-EF1-realized} establishes that all partial allocations in $\EFone(i,\values_i)$ can possibly be returned by the mechanism if agent $i$ reports $\values_i$. Therefore, the problem of whether Mechanism~\ref{mechanism:NOM+PO+EF1} is obviously manipulable (partially) reduces to whether some set $\EFone(i,\values'_i)$ has a better worst-case outcome than $\EFone(i,\values_i)$, with respect to the valuation vector $\values_i$. Before describing the subsequent lemma which develops this idea, we define some relevant notations. Given any set of (partial) allocations $S \subseteq \Pi_n(\items)$, an agent $i \in \agents$ and a valuation vector $\values$ let $\ell(i,\values, S) = \min_{\Alloc \in S} \ \values(A_i)$.

\begin{lemma}
\label{lemma:least-valued-EF1}
For any agent $i \in \agents$ and any pair of valuations $\values_i,\values_i'$, it holds that $\ell(i,\values_i, \EFone(i,\values_i)) \geq \ell(i,\values_i, \EFone(i,\values_i'))$.
\end{lemma}
\begin{proof}
Let $\Alloc = (A_1, A_2, \ldots, A_n) \in \EFone(i,\values_i)$ be an allocation such $\values_i(A_i) = \ell(i,\values_i, \EFone(i,\values_i))$, i.e. $\Alloc$ is the worst-case outcome for agent $i$ when her valuation vector is $\values_i$. If $\Alloc \in \EFone(i,\values_i')$, the lemma immediately follows, since $min_{\Alloc' \in \EFone(i,\values_i')} \values_i(A'_i)$ will certainly be at most $\values_i(A_i)$. Therefore, assume without loss of generality that $\Alloc \notin \EFone(i,\values_i')$. We will construct an allocation $\Alloc'$ with the following properties: $(i)$ $\values_i(A'_i) \leq \values_i(A_i)$, and $(ii)$ $\Alloc' \in \EFone(i,\values_i')$. These two properties together imply the desired inequality, since
$\ell(i,\values_i,\EFone(i,\values_i)) = \values_i(A_i) \geq \values_i(A'_i) \geq \ell(i,\values_i,\EFone(i,\values_i'))$, where the first inequality follows from $(i)$ and the second inequality follows from $(ii)$.

Let $S^* \subseteq \agents \setminus \{i\}$ be the subset of agents with non-empty bundles in $\Alloc$, i.e., $S^* \coloneqq \{j \in \agents \setminus \{i\} \ | \ A_j \neq \emptyset \}$. Define agent $j^* \coloneqq \argmax_{k \in S^*} \min_{g \in A_k} \values_i'(A_k \setminus \{g\})$ to be the agent in $S^*$ whose allocation ``up to one item'' has maximum value with respect to $\values_i'$, and let 
$g^* \coloneqq \argmax_{g \in A_{j^*}} \values_i(\{g\})$ be the favorite good with respect to $\values_i$  that $j^*$ has. Given these definitions, $\Alloc'$ is defined as follows. $A'_i \coloneqq A_{j^*} \setminus \{g^*\}$, $A'_{j^*} \coloneqq A_i \cup \{g^*\}$, and for all agents $k \in \agents \setminus \{i,j^*\}$ we have $A'_k \coloneqq A_k$. It remains to show that $\Alloc'$ satisfies $(i)$ and $(ii)$.

Towards proving $(i)$ we have $\values_i(A'_i) = \values_i(A_{j^*} \setminus \{g^*\}) \leq \values_i(A_i)$, where the last inequality follows from the facts that $\Alloc \in \EFone(i,\values_i)$ and that $g^*$ is the item with maximum value for $\values_i$ in $A_{j^*}$. 

Towards proving $(ii)$, note that $\agents \setminus \{i\}$ can be written as the union of three disjoint sets $S_1$, $S_2$ and $S_3$, where $S_1 = \agents \setminus (S^* \cup \{i\})$, $S_2 = (S^* \setminus \{j^*\})$, and  $S_3 = \{j^*\}$. We will show that in allocation $\Alloc'$ agent $i$ with valuation vector $\values_i'$, does not envy, up to one item, any of the agents in $S_1$, $S_2$, or $S_3$; this establishes that $\Alloc' \in \EFone(i,\values_i')$.
First, for every agent $k \in S_1$ we have that $A'_k = A_k = \emptyset$, by the definition of $S^*$, so agent $i$ cannot envy any such agent.
Second, given the definition of agent $j^*$, we know that for each agent $k$ in $S_2 = S^* \setminus \{j^*\}$, we have
\begin{align*}
\min_{g \in A'_k} \values_i'(A'_k \setminus \{g\}) & = \min_{g \in A_k} \values_i'(A_k \setminus \{g\}) \tag{$A'_k = A_k$ for all agents except $i$ and $j^*$}\\
& \leq \min_{g \in A_{j^*}} \values_i'(A_{j^*} \setminus \{g\}) \tag{by the definition of $j^*$} \\
& \leq \values_i'(A_{j^*} \setminus \{g^*\}) \tag{$g^* \in A_{j^*}$}\\
& = \values_i'(A'_i). \tag{$A'_i = A_{j^*} \setminus \{g^*\}$}
\end{align*}

Therefore, agent $i$ with valuation $\values_i'$ does not envy up to one item (and, in fact, up to any item) any agent in $S_2$.

Third, recall that $\Alloc \notin \EFone(i,\values_i')$. Moreover, by definition, agent $j^*$ has the bundle that $\values_i'$ likes the most ``up to one item''. Thus, agent $i$ envies $j^*$ even after the removal of \emph{any} one item, and, in particular, $\values'_i(A_i) < \values'_i(A_{j^*} \setminus \{g^*\})$. Using this, we have
\begin{align*}
    \values_i'(A'_i) & = \values_i'(A_{j^*} \setminus \{g^*\}) \tag{$A'_i = A_{j^*} \setminus \{g^*\}$} \\
    & > \values_i'(A_i) \tag{$\Alloc \notin \EFone(i,\values_i')$} \\
    & = \values_i'(A_i \cup \{g^*\} \setminus \{g^*\}) \\
    & = \values_i'(A'_{j^*} \setminus \{g^*\}) \tag{$A'_{j^*} = A_i \cup \{g^*\}$}.
\end{align*}
That is, agent $i$ with valuation $\values_i'$ does not envy, up to one item, any agent in $S_3 = \{ j^* \}$. 

This establishes that $\Alloc' \in \EFone(i,\values_i')$ and concludes the proof.
\end{proof}



Given the previous lemmas, we can show our main theorem with respect to not obviously manipulability.

\begin{theorem}\label{thm: EF1 implies NOM}
If agents' valuation functions are additive and $\mathcal{M}^*$ outputs $\EFone$, non-wasteful and clean allocations, then Mechanism~\ref{mechanism:NOM+PO+EF1} is not obviously manipulable.
\end{theorem}

\begin{proof}
We first show that agents cannot improve their worst-case utilities by misreporting.

Let $i \in \agents$ be an agent whose true valuation is $\values_i$ and let her misreport be $\bids_i$. By Lemma~\ref{lemma:every-EF1-realized}, for every reported valuation vector $\bids$ of agent $i$, 
and for every allocation $\Alloc \in \EFone(i,\bids)$, there exists a valuation profile $\values_{-i}$ such that $\Alloc$ is the output of Mechanism~\ref{mechanism:NOM+PO+EF1}.
On the other hand, by Lemma~\ref{lemma:mech is EF1}, Mechanism~\ref{mechanism:NOM+PO+EF1} always outputs allocations $\Alloc$ that are envy free up to one item, clean, and non-wasteful. 
Therefore, every output allocation of Mechanism~\ref{mechanism:NOM+PO+EF1} is in $\EFone(i,\bids)$.
Therefore, if agent $i$ reports $\values_i$ (respectively $\bids_i$), then the set of all possible allocations of Mechanism~\ref{mechanism:NOM+PO+EF1} is exactly $\EFone(i,\values_i)$ (respectively $\EFone(i,\bids_i)$).
Overloading notation, let $A^*_i(\values_i, \values_{-i})$ be the allocation of agent $i$ in Mechanism~\ref{mechanism:NOM+PO+EF1}, on input $(\values_i, \values_{-i})$.
For the worst-case utility of agent $i$ we have
\begin{align*}
    \min\limits_{\values_{-i}} \values_i( A^*_i(\values_i, \values_{-i})  )& = \min\limits_{ \Alloc \in \EFone(i,\values_i)} \values_i(A_i)\\
    & = \ell(i, \values_i, \EFone(i,\values_i)) \tag{by definition} \\
    & \geq \ell(i, \values_i, \EFone(i, \bids_i)) \tag{Lemma \ref{lemma:least-valued-EF1}}\\
    & = \min\limits_{\values_{-i}} \values_i( A^*_i(\bids_i, \values_{-i})  ).
\end{align*}

This establishes Inequality~\ref{inequality:NOM-worst-case} from the definition of NOM. To conclude the proof, it remains to show that the best-case utility of an agent is also not improved by misreporting. Note that the best case for agent $i \in \agents$ occurs when all other agents have no value for the items, i.e. $D_j = \emptyset$ for all $j \in \agents \setminus \{i\}$. Hence the sets $\{D_k\}_{k \in [n]}$ are pairwise-disjoint, and Mechanism~\ref{mechanism:NOM+PO+EF1} (via \emph{Case \rm{I}}) allocates the entirety of $D_i$ to agent $i$, which is impossible to improve upon.
\end{proof}

The proof of the above theorem along with the supporting lemmas establish Theorem \ref{theorem:NOM+PO+EF1}.




%
%

\subsection{Ex-ante EF, Ex-post fPO, Ex-post EF1 and NOM}

In the previous section we established (Application~\ref{application: additive}) the existence of a deterministic NOM mechanism that outputs ex-post fPO and EF1 allocations for additive agents. Here we make a short remark that it is not possible to improve upon this result by adding ex-ante fairness guarantees. Specifically, there does not exist a randomized mechanism that is NOM in expectation and outputs ex-post fPO and EF1 allocations, that is additionally ex-ante envy-free. This follows directly from the following impossibility result of~\cite{freeman2020best}.

\begin{theorem}[\cite{freeman2020best}]
There exists instances with additive valuation where there is no randomized allocation that is simultaneously ex-post fPO and $\EFone$, and ex-ante envy-free.
\end{theorem}

\section{Best of Both Worlds}\label{sec: best of both}

In the previous sections we saw how fairness guarantees can be turned into NOM guarantees. For example, Lemma~\ref{lemma:ex-ante-EF-and-NOM} says that the worst-case guarantee of NOM is satisfied for ex-ante proportional algorithms, while Theorem~\ref{theorem:NOM+PO+EF1} turns any $\EFone$ algorithm into a NOM (plus $\EFone$) mechanism. In this section we show that this connection can be exploited in the other direction to prove impossibility results for fair algorithms. Specifically, we show that certain ``best-of-both-worlds'' fairness guarantees, that is, randomized allocations that satisfy an ex-ante and ex-post guarantees simultaneously, are not guaranteed to exist.
Our impossibility result uses the following theorem.

\begin{theorem}
\label{theorem:NOM-negative-result-combo}
Every deterministic (or randomized) mechanism for $n=2$ agents with additive and normalized utilities, that always outputs allocations that (ex-post) maximize the number of agents having positive utility, is obviously manipulable.
\end{theorem}

The proof is deferred to Appendix~\ref{app:missing}. We note that we can apply Theorem~\ref{theorem:NOM-negative-result-combo} to show that for $n=2$ there is no NOM mechanism that outputs (ex-post) and MNW solution, since MNW is scale free (and therefore normalization comes without loss of generality).

We are now ready to prove our main result for this section.

\begin{theorem}\label{thm: bobw impossibility}
Under additive valuations, randomized allocations that are ex-ante proportional, ex-post Pareto efficient and ex-post maximize the number of agents having positive utility do not exist, for any number of agents and any number of items.
\end{theorem}
\begin{proof}
We begin by establishing the result for $n=2$ agents. Towards a contradiction, assume that for additive valuations, randomized allocations that are ex-ante proportional, ex-post Pareto efficient and ex-post maximize the number of agents having positive utility always exits, and let $\mathcal{R}$ be a randomized mechanism that, given the reported valuations of agents, outputs such a allocation. Note that this mechanism need not be computationally efficient. We will show that $\mathcal{R}$ is NOM in expectation, in direct contradiction to Theorem~\ref{theorem:NOM-negative-result-combo}.


Since, $\mathcal{R}$ outputs ex-ante proportional allocations, the worst-case guarantee of NOM (inequality~\ref{inequality:NOM-worst-case}) holds, via Lemma~\ref{lemma:ex-ante-EF-and-NOM}. Additionally, the best-case guarantee (inequality~\ref{inequality:NOM-best-case}) is also satisfied: if an agent $i \in \items$ reports her true valuation, then in the best case, all other agents only desire her least-valued item; in this case, agent $i$ receives all items except her least-valued item. This outcome cannot be (strictly) improved as agent $i$ must lose at least one item: the mechanism maximizes the number of agents having positive utility and each agent reports having a positive utility of at least one item, therefore, all items cannot be allocated to the same agent. Therefore, $\mathcal{R}$ is NOM is expectation, a contradiction. 

Thus, for $n=2$ agents randomized allocation that are ex-ante proportional, ex-post Pareto efficient, and ex-post maximize the number of agents having positive utility do not always exist. In particular, there is a $2$ agent instance where such randomized allocations do not exist. This instance can be easily converted into a $n>2$ agent instance by adding dummy agents $a_1, a_2, \ldots, a_k$ and dummy items $j_1, j_2, \ldots, j_k$ such that $(i)$ for each $i \in [k]$, agent $a_i$'s value for item $j_1$ is $1$ and her value for every other item is $0$, $(ii)$ the existing agents have zero value for each dummy item. In this new instance, every Pareto efficient allocation must be such that agent $a_i$ is allocated item $j_i$ and no other item that existing agents value positively. Thus, a randomized allocation that is ex-ante proportional, ex-post Pareto efficient, and ex-post maximizes the number of agents having positive utility would continue to satisfy these properties after removing the dummy agents and dummy items. However, as already established, such allocations do not exist for $n=2$ agents. This concludes the proof of Theorem~\ref{thm: bobw impossibility}.
\end{proof}

The existence of ex-ante proportional, ex-post Pareto efficient and ex-post $\EFone$ allocations remains an elusive open problem. Our result shows that replacing ex-post $\EFone$ with a different, mild fairness guarantee, namely ex-post maximizing the number of agents with positive utility, is impossible. The following corollary is immediate (and we note that point \#3 in the following corollary is shown by~\cite{freeman2020best}).


\begin{corollary}
Under additive valuations, randomized allocations that satisfy the following properties do not always exist.
\begin{enumerate}
    \item ex-ante proportional, ex-post Pareto efficient and ex-post egalitarian welfare maximizing
    \item ex-ante proportional and ex-post leximin allocations
    \item ex-ante proportional and ex-post MNW.
\end{enumerate}
\end{corollary}

\bibliographystyle{alpha}
\bibliography{main}

\appendix
\section{The PS-Lottery Algorithm of [Aziz, 2020b]}
\label{appendix:PS-Lottery}

A square matrix $M \in [0,1]^{k \cdot k}$ is \emph{bistochastic} iff the sum of entries in each of its rows and columns is equal to one, i.e., for each $i \in [k]$, we have $\sum_{j = 1}^k M_{i,j} = \sum_{j = 1}^k M_{j,i} = 1$. Additionally, a bistochastic matrix $N \in \{0,1\}^{k \cdot k}$ is a \emph{permutation matrix} --- each row and column of a permutation matrix contains exactly one entry having value one and all other entries are zero.

In essence, the PS-Lottery algorithm of \cite{aziz2020simultaneously} is based on the following two well-known algorithms:\\
\noindent
\emph{Birkhoff's algorithm.} Given a bistochastic matrix $M \in [0,1]^{k \cdot k}$ as input, Birkhoff's algorithm can be used to decompose, in polynomial time, the matrix $M$ into a convex combination of permutation matrices. That is, Birkhoff's algorithm outputs permutation matrices $\{M_i\}_{i=1}^t$ and positive real numbers $\{p_i\}_{i=1}^t$ such that $M = \sum_{i=1}^k p_i M_i$ and $\sum_i p_i = 1$; here $t = \mathcal{O}(k^2)$ is a positive integer.

\noindent
\emph{Probabilistic serial algorithm.} Given a fair division instance wherein agents have additive valuations, the probabilistic algorithm outputs a fractional allocation that is envy-free (and hence proportional). The fractional allocation output by probabilistic serial can be interpreted as the output of the following continuous procedure: starting from time $t = 0$, simultaneously, each agent start consuming items in the order of their preference (i.e., if $v_{i,j_1} \geq v_{i,j_2} \ldots, \geq v_{i,j_m}$ then the items are consumed in the order $j_1, j_2, \ldots, j_m$\footnote{ties can be broken arbitrarily}) and at a rate of one item per unit time. If an item is fully consumed, then agents start consuming the next item as per their preference order. The algorithm terminates when all items have been consumed, which happens at time $t = \frac{m}{n}$. As a direct consequence of this, $(i)$ each agent gets exactly $\frac{m}{n}$ fraction of items at the end, and $(ii)$ the resultant fractional allocation is envy-free, since at every point agents are consuming their most-valued remaining item.

The PS-Lottery algorithm uses Birkhoff's algorithm to decompose the fractional allocation output by probabilistic serial into a convex combination over integral allocations, i.e., a randomized allocation. In addition, each integral allocation in the support of the randomized allocation is EF1. 

In Section~\ref{sec: fair}, towards showing that the PS-Lottery algorithm is NOM in expectation, we use the following properties of the PS-Lottery algorithm.
\begin{property}
The randomized allocation output by the PS-Lottery algorithm is ex-ante EF and ex-post EF1.
\end{property}

\begin{property}
\label{PS-Lottery-property-2}
The total fraction of items that each agent gets in the expected fractional allocation returned by the PS-Lottery algorithm is $\frac{m}{n}$
\end{property}

Additionally, suppose that agent $i \in \agents$ prefers the items in the order $v_{i, 1} \geq v_{i, j_2} \ldots \geq v_{i, j_m}$, and the preference order to all the other agents is opposite, i.e., for all agents $k \in [n] \setminus \{i\}$, we have $v_{k, m} \geq v_{k,m-1} \ldots \geq v_{k,1}$, then in the fractional allocation output by probabilistic serial, agent $i$ gets the items $1, 2, \ldots, \lfloor \frac{m}{n} \rfloor$ entirely, and a fraction $\frac{m}{n} - \lfloor \frac{m}{n} \rfloor$ of the item $\lfloor \frac{m}{n} \rfloor + 1$.




\section{Missing proofs}\label{app:missing}

\begin{proof}[Proof of Theorem~\ref{thm: NOM + egalitarian + n agents}]
Consider any (randomized or deterministic) mechanism $\mathcal{M}$ that maximizes the minimum utility among agents. To show that $\mathcal{M}$ is obviously manipulable, we consider an instance with $n=3$ agents and $m=4$ items. Let the true (normalized) valuation of agent $1$ be $\values_1 = (0.3, 0.3, 0.3, 0.1)$.

First, we will show that if agent $1$ reports honestly, then the worst-case case outcome is that she is allocated only item $4$, i.e., her utility in the worst-case is $0.1$. Note that since $\mathcal{M}$ maximizes the number of agents with positive utility, and agent $1$ has a positive value for all $4$ items, it must allocate at least one item to agent $1$; otherwise either agent $2$ or $3$ has more than one item, one of which could be transferred to agent $1$ to make her utility positive. Now, consider the case where the reports of other agents are $\bids_2 = (0, 0, 1, 0)$ and $\bids_3 = (0.05, 0.05, 0.9, 0)$. Here, item $3$ must be allocated to agent $2$ because it is the only item that she desires. By a straightforward case analysis, the unique way in which we can allocate the remaining items (items $1$, $2$ and $4$) to maximize the minimum utility is by allocating item $4$ to agent $1$ and items $1$ and $2$ to agent $3$. Since this allocation is unique, it must be the output of $\mathcal{M}$, irrespective of whether it is deterministic or randomized.
Therefore, the utility of agent $1$ is $0.1$ under honest reporting, in the worst case.


Now, consider the case where agent $1$ reports $\bids_1 = (\frac{1}{3}, \frac{1}{3}, \frac{1}{3}, 0)$. Since $\mathcal{M}$ maximizes the number of agents having non-zero utility, and agent $1$ positively values the first three items, she must get at least one of them: otherwise some other agent will have at least two of them, and transferring one of them to agent $1$ would strictly increase the number of agents with non-zero utility. Hence, the worst-case utility of agent $1$ when reporting $\bids_1$ but her true values are $\values_1$ is $0.3$, her true value for each of the first three items. 
Since the worst-case utility improved, $\mathcal{M}$ is obviously manipulable.

The same construction can be extended to $n > 3$ agents and $m = n+1$ items by considering the case wherein agent $1$'s true valuation $\values_1 = (\frac{1}{n} - \epsilon, \frac{1}{n} - \epsilon, \frac{1}{n} - \epsilon, \ldots, \frac{1}{n} - \epsilon, n. \epsilon)$ for any positive constant $\epsilon < \frac{1}{n(n+1)}$, and she misreports to $\bids_1 = (\frac{1}{n}, \frac{1}{n}, \frac{1}{n}, \ldots, \frac{1}{n}, 0)$ leading to an improvement in her worst-case utility. 
\end{proof}

\begin{proof}[Proof of Theorem~\ref{theorem:NOM-negative-result-combo}]
Let $\mathcal{M}$ be such a deterministic (or randomized) mechanism. We will prove the theorem (i.e., show that $\mathcal{M}$ is obviously manipulable) for the case of $m=2$ items; the proof easily generalizes to more than $2$ items.

Consider the case when both the agents report a value of $1$ for the first item and zero for the other item, i.e. $\bids_{i} = (1,0)$ for $i \in \{ 1, 2 \}$.
Given these reports, there must be an agent who gets item $1$ with probability at least $\frac{1}{2}$ (if $\mathcal{M}$ is deterministic then this probability will be exactly $1$); assume that this is agent $1$, without loss of generality.

Consider the case that the true valuation of agent $1$ is $\values_1 = (\frac{2}{3}+\epsilon, \frac{1}{3} - \epsilon)$, for some small constant $\epsilon > 0$. The worst case for agent $1$ under honest reporting happens when the reported valuation of agent $2$ $\bids_2 = (1,0)$, in which case item $1$ must be allocated to agent $2$ and item $2$ to agent $1$ -- irrespective of whether the mechanism $\mathcal{M}$ is deterministic or randomized since $\mathcal{M}$ maximizes the number of agents having positive utility \emph{ex-post}, and allocating item $1$ to agent $2$ and item $2$ to agent $1$ is the unique allocation that gives both agents non-zero utility. Therefore, the worst-case utility of agent $1$ is at most $\frac{1}{3} - \epsilon$.

Now, consider the dishonest report $\bids_1 = (b_{1,1}, b_{1,2}) = (1,0)$. If $b_{2,2} > 0$, then mechanism $\mathcal{M}$, in order to maximize the number of agents having positive utility, will allocate item $2$ to agent $2$ and item $1$ to agent $1$, i.e., agent $1$'s utility in this case (per $\values_1$) would be $\frac{2}{3}+\epsilon$. On the other hand, if $b_{2,2} = 0$ (and therefore $b_{2,1} = 1$) then as per our choice of mechanism $\mathcal{M}$, agent $1$ will receive item $1$ with probability at least $\frac{1}{2}$. Hence, her utility in this case would be at least $\frac{1}{2} \values_{1,1} = \frac{1}{3} + \frac{\epsilon}{2}$. The above case analysis shows that if agent $1$ reports valuation $\bids_1$ then her worst-case utility will be $\frac{1}{3} + \frac{\epsilon}{2}$, which is strictly larger then her worst-case utility when reporting true valuation.
Therefore, $\mathcal{M}$ is obviously manipulable.
\end{proof}










\end{document}